\documentclass{article}
\usepackage[margin=1in]{geometry}

\usepackage{textcomp}
\usepackage[normalem]{ulem}

\DeclareRobustCommand{\msg}[1]{%
    \ifmmode
        \mathord{\mbox{\raisebox{.5pt}{\textcircled{\raisebox{-.9pt}{#1}}}}}%
    \else
        \raisebox{.5pt}{\textcircled{\raisebox{-.9pt}{#1}}}%
    \fi
}
\newcommand{\dpow}{DPoW\xspace}

\usepackage{pdfpages}
\usepackage{listings}
\usepackage{xcolor}
\usepackage{bbm}
\usepackage{verbatim}
\usepackage{caption}

\usepackage{subcaption}
\usepackage{xspace}
\usepackage{comment}
\usepackage{algorithm}
\usepackage{amsmath,amsthm,amssymb}
\usepackage[noend]{algpseudocode}
\algrenewcommand{\algorithmiccomment}[1]{\hfill\ensuremath{//} #1}
\makeatletter
\newcommand{\LeftComment}[1]{%
  \Statex \hskip\ALG@tlm // #1%
}
\makeatother
\algrenewcommand{\Call}[2]{\textbf{call}\ \textsc{#1}(#2)}
\usepackage{enumitem}

\usepackage{graphicx}
\usepackage{import}
\usepackage{tikz}

\usepackage{thmtools}
\usepackage{thm-restate}
\usepackage[hidelinks]{hyperref}
\usepackage{cleveref}

\definecolor{codegreen}{rgb}{0,0.6,0}
\definecolor{codegray}{rgb}{0.5,0.5,0.5}
\definecolor{codepurple}{rgb}{0.58,0,0.82}
\definecolor{backcolour}{rgb}{0.95,0.95,0.92}

\newtheorem{definition}{Definition}
\newtheorem{theorem}{Theorem}
\newtheorem*{theorem*}{Theorem}
\newtheorem{lemma}{Lemma}
\crefname{lemma}{lemma}{lemmas}
\Crefname{lemma}{Lemma}{Lemmas}

\newtheorem{property}{Property}
\crefname{property}{property}{properties}
\Crefname{property}{Property}{Properties}

\newtheorem{assumption}{Assumption}
\crefname{assumption}{assumption}{assumptions}
\Crefname{assumption}{assumption}{assumptions}

\crefname{guarantee}{guarantee}{guarantees}
\Crefname{guarantee}{guarantee}{guarantees}

\usepackage{authblk}

\begin{document}

\title{Fast Deterministically Safe Proof-of-Work Consensus}

\author[1,*]{Ali Farahbakhsh}
\author[2,*]{Giuliano Losa}
\author[1,*]{Youer Pu}
\author[1]{Lorenzo Alvisi}
\author[3]{Ittay Eyal}
\affil[1]{Cornell University}
\affil[2]{Stellar Development Foundation}
\affil[3]{Technion}
\date{} %

\maketitle

\begin{abstract}

Permissionless blockchains achieve consensus while allowing unknown nodes to join and leave the system at any time.
They typically come in two flavors: proof of work (PoW) and proof of stake (PoS), and both are vulnerable to attacks.
PoS protocols suffer from long-range attacks, wherein attackers alter execution history at little cost, and PoW protocols are vulnerable to attackers with enough computational power to subvert execution history.
PoS protocols respond by relying on external mechanisms like social consensus; PoW protocols either fall back to probabilistic guarantees, or are slow.

We present \emph{Sieve-MMR}, the first fully-permissionless protocol with deterministic security and constant expected latency that does not rely on external mechanisms.
We obtain Sieve-MMR by porting a PoS protocol (MMR) to the PoW setting. From MMR we inherit constant expected latency and deterministic security, and proof-of-work gives us resilience against long-range attacks.
The main challenge to porting MMR to the PoW setting is what we call \emph{time-travel attacks}, where attackers use PoWs generated in the distant past to increase their perceived PoW power in the present.
We respond by proposing \emph{Sieve}, a novel algorithm that implements a new broadcast primitive we dub \emph{time-travel-resilient broadcast} (TTRB).
Sieve relies on a black-box, \emph{deterministic} PoW primitive to implement TTRB, which we use as the messaging layer for MMR.

\end{abstract}

\section{Introduction}\label{sec: introduction}

\begingroup
\renewcommand{\thefootnote}{*}
\footnotetext{These author contributed equally and are given in alphabetical order.}
\endgroup

Cryptocurrencies like Bitcoin and smart-contract platforms such as Ethereum aim to provide universal decentralized access to services like payments, banking, insurance, and e-commerce.
They aspire to present users worldwide with an open-access, secure, and \emph{fast} transaction log.
Typically, a total-order broadcast (TOB) protocol implements the log abstraction, and participation in the protocol is permissionless thanks to the use of proof-of-stake (PoS) or proof-of-work (PoW).
In PoW, participation requires solving expensive computational puzzles; in PoS, it requires putting cryptocurrency in escrow.

PoW and PoS paradigms are vulnerable to attacks that, though paradigm-specific,  share a similar goal: crafting an alternate execution history to confuse newly joining nodes who did not witness the past execution firsthand. No existing protocol, in either paradigm, has satisfactorily addressed these attacks.
{\em Long-range attacks} are an instance of these attacks specific to PoS systems: they consist in purchasing keys belonging to former participants---likely cheaply, as those parties no longer have a skin in the game---and using those keys to fabricate an alternate execution history. %
Without mechanisms external to the system (\emph{e.g.}, secure checkpoints, social consensus, etc.), there exists no defense against such attacks~\cite{tasBitcoinEnhancedProofofStakeSecurity2023,azouviPikachuSecuringPoS2022}.
PoW systems are not vulnerable to long-range attacks; however, they are either prohibitively slow~\cite{Sandglass, Gorilla} or rely on proof-of-work puzzles with probabilistic guarantees, which leave open the possibility that a lucky attacker may successfully create an alternate history ({\em e.g.}, in Bitcoin, a fork of the longest chain).
Even if in practice the failure probability can be made sufficiently small, proving safety for these probabilistic protocols is quite tricky; for instance, it took the community considerable effort to establish Bitcoin's security (\emph{e.g.},~\cite{garayBitcoinBackboneProtocol2015,demboEverythingRace2020}).

We present Sieve-MMR, the first permissionless \emph{PoW} TOB protocol with \emph{deterministic security} and \emph{constant expected latency}.
Sieve-MMR relies on a cryptographic hash function to obtain a \emph{deterministic proof-of-work primitive} \dpow.
Generating a \dpow requires a fixed number of hash computations, and we assume that adversaries control a minority of the computation power.
\dpow is similar to a verifiable delay function~\cite{bonehVerifiableDelayFunctions2018}, but without the requirement that the computation steps be performed serially.

A key principle guides the design of Sieve-MMR: \emph{decoupling the consensus logic from the logic used to control the undesirable side-effects of permissionless participation}. This separation of concerns opens an intriguing new possibility: safely deploying existing low-latency consensus protocols, originally developed under stronger models, in a fully permissionless setting.

Accordingly, we design Sieve-MMR in two layers. The top layer implements consensus using a protocol inspired by what we call the MMR protocol~\cite[Appendix A]{malkhiByzantineConsensusFully2023}, a fast and deterministically safe protocol in the family of dynamically available~\cite{lewis-pyePermissionlessConsensus2024} PoS TOB protocols.
Our main technical contributions are in the bottom layer.
For the first time, this layer specifies and implements the message delivery guarantees that dynamically available protocols like MMR must rely on to maintain correctness in a fully-permissionless setting.

\par{\bf Time travel considered harmful.} Embedding dynamically available protocols in a fully-permissionless setting exposes them to an insidious new threat.
Not only must they defend against adversaries using their current computing power to modify past consensus decisions, but also against adversaries leveraging their {\em past} computing power to alter consensus decisions in the present!

This vulnerability stems from how dynamically available protocols implement consensus~\cite{gafni_brief_2023, malkhi_towards_2023}:  they build upon traditional quorum intersection arguments, which in turn rely on correct ({\em i.e.}, protocol-abiding) nodes generating a sufficiently strong majority of the messages being sent.
However, in a permissionless setting, nothing prevents corrupted nodes from generating messages at some point in the past, holding onto them, and sending them as if they were generated in the present, thereby distorting the quorums correct nodes perceive.

\par {\bf TTRB and Sieve.} 
We capture the messaging properties that MMR and similar protocols need to be immune to such {\em time travel attacks} with a new broadcast primitive: {\em time-travel-resilient broadcast} ({\em TTRB}).
TTRB operates in rounds and provides two guarantees: ($i$) all messages TTRB-delivered in a given round were generated in the previous round; and ($ii$) all messages that correct nodes generated in the previous round are TTRB-delivered by correct nodes in the current round. 

We implement TTRB with the novel protocol {\em Sieve}.
Like its namesake, we use Sieve to filter out ``impurities''---in our case, time-traveling messages.
Sieve limits the number of messages that a node can generate in a given round by augmenting each message with the \dpow evaluation of its payload.
Furthermore, each message in Sieve is associated with a timestamp, which intuitively represents the round in which the message generation began.
Messages carry this timestamp as an attribute, and in any given round, Sieve should only deliver messages that are timestamped from the previous round.
Of course, an adversary might attempt to counterfeit the timestamp attribute, trying to pass off an earlier message as a later one.
To counter this, Sieve implements a Byzantine-tolerant mechanism that can identify and discard messages with counterfeit timestamps.

At the core of this mechanism is the DAG of \dpow evaluations induced by requiring correct nodes to logically include in each message they generate a ``coffer'' containing all messages they accepted in the previous round.
By iteratively analyzing and pruning the DAG, Sieve is able to tell when a message \(m\) claiming a generation time \(s\) includes at least one message generated by a correct node at time \(s-1\); inductively, this guarantees that \(m\) was generated no earlier than time \(s\), and Sieve filters out all messages that do not pass this test.

Concretely, Sieve comprises two filtering policies: Bootstrap-Sieve and Online-Sieve.
Correct nodes execute Bootstrap-Sieve upon joining the execution, in order to catch up; once caught up, they can switch to Online-Sieve.
Bootstrap-Sieve operates over the entire DAG of \dpow evaluations, using the full prior history to inform its analysis. 
Online-Sieve is instead much cheaper: it requires only a set of filtered messages from the previous messaging round.

\par {\bf Adversarial assumptions.} Sieve's guarantees hold when adversaries are collectively 1/2-bounded, {\em i.e.}, when over any sufficiently long stretch of time---long enough for a correct node to compute at least one \dpow---attackers compute strictly less than half of the total number of \dpow evaluations.
This requirement is similar to the common PoW majority assumption.
The MMR protocol, on the other hand, tolerates only a 1/3-bounded adversary, and Sieve-MMR inherits this stronger assumption. 

\par {\bf Determinism.}
Sieve-MMR's claims of deterministic safety are qualified: they apply within the confines of a Dolev-Yao-style model~\cite{dolev-yao} in which adversaries do not break cryptographic primitives and do not guess messages.
Unlike Sieve, the guarantees of PoW protocols like Bitcoin remain probabilistic even if adversaries do not break cryptographic assumptions, as they rely on a non-deterministic process that naturally lends itself to a stochastic analysis. As noted before, establishing rigorously the security of these protocols is notoriously difficult.

\par {\bf Practical limitations.}
While Sieve-MMR marks a significant step  towards practical PoW protocols that can deliver deterministic TOB with constant-latency, some key hurdles remain.
Notably, Sieve-MMR assumes a synchronous network and relies on all-to-all broadcast in every round, incurring bandwidth costs that scale quadratically with the number of active nodes.
Moreover, running Bootstrap-Sieve requires solving an exponential-time problem over a graph that captures the protocol's execution thus far. Although nodes are not bound to complete this computation within a fixed time, they must do so before they can actively participate in the protocol.

\par {\bf Summary of the contributions.}
In conclusion, we make the following contributions:
\begin{itemize}
    \item We introduce time-travel-resilient broadcast (TTRB), a new broadcast abstraction that guarantees a messaging layer immune to time-travel attacks.
    \item We derive Sieve, a new algorithm that correctly implements TTRB assuming that attackers control a minority of the computation power in the system.

    \item We present Sieve-MMR, the first consensus protocol for the permissionless setting that, leveraging the guarantees of the Sieve-enabled TTRB primitive, achieves constant expected latency and deterministic safety without trusting external mechanisms.
\end{itemize}

PlusCal/TLA+ specifications of the Sieve and MMR algorithms can be found online~\cite{losaFormalModelsSieveMMR2025} and in~\Cref{tla-specs}. 

\section{Background}\label{sec: background}

Sieve and Sieve-MMR draw inspiration both from classic abstractions in  fault tolerant distributed computing and from prior work on permissionless consensus.
Three essential notions are useful to contextualize Sieve and Sieve-MMR within this broad landscape: \emph{total-order broadcast}, \emph{dynamically available protocols}, and the \emph{sleepy model}.

\textbf{Total-Order Broadcast.}
Total-order broadcast (TOB)~\cite{lamport_time_1978}, also known as {\em atomic broadcast}~\cite{cristian1995atomic}, condenses into a specification the agreement and progress aspects of State Machine Replication (SMR) ~\cite{lamport_time_1978,SMR}, the most general approach for building fault-tolerant distributed systems. SMR provides to its clients the abstraction of a single state machine that never fails by replicating the machine's state and coordinating the replicas actions. 

SMR uses the abstraction of a command log; each replica has one such log, and one copy of the state. If correct and deterministic replicas, starting from the same initial state, agree on the ordering of the client-issued commands within their logs,  then executing the logs up to any fixed index produces the same state, and the same reply to each command,  at each replica. Voting can then be used to ensure that clients  only process replies that a single correct deterministic replica, given the same initial state and command sequence, would have generated. 

To support this approach, the TOB specification assumes a set of nodes ($e.g.$, replicas) that receive messages ({\em e.g.}, {\em commands} in SMR) from clients, and broadcast the messages among themselves.
Total-order broadcast requires of all correct nodes to deliver messages consistently, $i.e.$, the sequence of messages delivered by any two correct  nodes ($e.g.$, the command logs for two replicas) should satisfy the prefix relation.
The system must also make progress infinitely often, $i.e.$, all messages are eventually delivered.

\textbf{Dynamically Available Protocols.} Total-order broadcast is a core technical challenge also for blockchain and decentralized computing platforms.
While the terminology may differ, the underlying concern remains the same: nodes must maintain a consistent view of the system state, which should infinitely often progress by incorporating client transactions.

A key distinction between traditional fault-tolerant systems and modern blockchains lies in the treatment of node availability. In classical settings, nodes are assumed to be either active or faulty. In contrast, blockchains introduce the notion of {\em node churn}, allowing nodes to become inactive voluntarily---even when they are not faulty.
Churn manifests in various forms, recently formalized through degrees of {\em permissionlessness}~\cite{ByzantineGeneralsPermissionless,lewis-pyePermissionlessConsensus2024}. At one end of the spectrum, nodes lack tangible identities and may join or leave the system unilaterally; at the other lies the traditional model, where participation is more tightly controlled and nodes have unique identities.

Dynamically available protocols---corresponding to the model formalized by Lewis-Pye and Roughgarden~\cite{ByzantineGeneralsPermissionless,lewis-pyePermissionlessConsensus2024}---occupy a middle ground on this spectrum. These protocols assume that nodes have unique identities and that the pool of participants is globally known, while still allowing nodes to become inactive at will. Joining the pool, however, requires explicit approval from its current members.

Inspired by Ethereum~\cite{pos-ethereum}, many dynamically available protocols have adopted probabilistic safety guarantees. More recently, a new class of protocols has emerged that achieves deterministic safety~\cite{losa_consensus_2023, gafni_brief_2023, malkhi_towards_2023, malkhiByzantineConsensusFully2023, momose_constant_2022}. Their central insight is that the core correctness argument behind traditional total-order broadcast---namely, quorum intersection---remains applicable in the dynamically available setting. These protocols typically follow a layered design: they first solve a variant of graded agreement~\cite{commit-adopt, feldman-micali-graded-broadcast, katz-gradecast}, and then leverage it as a black box to implement total-order broadcast. For this approach to apply, however, they must rely on a strong and often restrictive model.

\textbf{The Sleepy Model.} 
The sleepy model~\cite{pass_sleepy_2017} serves as the {\em de facto} standard assumed by dynamically available protocols. In this model, nodes are part of a public key infrastructure (PKI), and all participants are known to one another. Correct nodes may freely alternate between active and inactive states, while faulty nodes {\em are  perpetually active}. Crucially, the model requires that a majority of active nodes be correct at all times.

This model is restrictive: for example, it cannot capture scenarios where the system size grows while maintaining a constant fraction of Byzantine nodes. It also places a burden on correct nodes, which must establish a strong presence from the outset. Some dynamically available protocols inherit this limitation from the sleepy model~\cite{momose_constant_2022, gafni_brief_2023}, while others relax it by adopting a stronger majority assumption~\cite{malkhi_towards_2023} than what is standard in Bitcoin~\cite{nakamotoBitcoinPeertopeerElectronic2008}. Specifically, these protocols compare the number of correct nodes at a given time with the number of faulty nodes over a time interval. This interval always begins at the start of execution, which means that the number of correct nodes at some time~$t$ must exceed the total number of Byzantine nodes over some interval~$[0, t']$ for~$t\leq t'$.While this assumption allows for some fluctuation in Byzantine participation, it remains overly strong.

\textbf{Time Travel is Harmful.} Weakening the majority assumption to exclude the entire execution history exposes dynamically available protocols to time-travel attacks. In such scenarios, Byzantine nodes can resurface messages from the distant past, distorting the quorum views of correct nodes in the present. This undermines the effectiveness of the majority assumption, rendering it essentially useless.

\Cref{fig:time-travel} illustrates a time-travel attack. Nodes  \(n_1\), \(n_2\), and \(n_3\) are correct;  nodes \(n_4\) and \(n_5\) are Byzantine. Time progresses from left to right and is divided into two protocol steps. Circles represent messages ({\em e.g.}, votes in a consensus algorithm), and arrows indicate when and where those messages are delivered. In step 0, all nodes are active; 
 in step~1, only \(n_1\), \(n_2\), and \(n_4\) remain active. Notably, in both steps correct active nodes outnumber Byzantine active nodes.
 
However, even assuming authenticated messages, the adversary controlling the Byzantine nodes can delay the delivery of \(n_5\)'s step-0 messages until step~2 (possibly by having \(n_5\) forward its messages or share its signing keys with \(n_4\)). By so doing, the adversary causes correct nodes to perceive a distorted quorum in step 2: a strict majority of correct nodes appears to no longer be present! In consensus protocols that rely on quorum intersection for safety, such distortions can lead to  violations of safety guarantees.
\begin{figure}[ht]
    \centering
    \includegraphics[width=0.5\columnwidth]{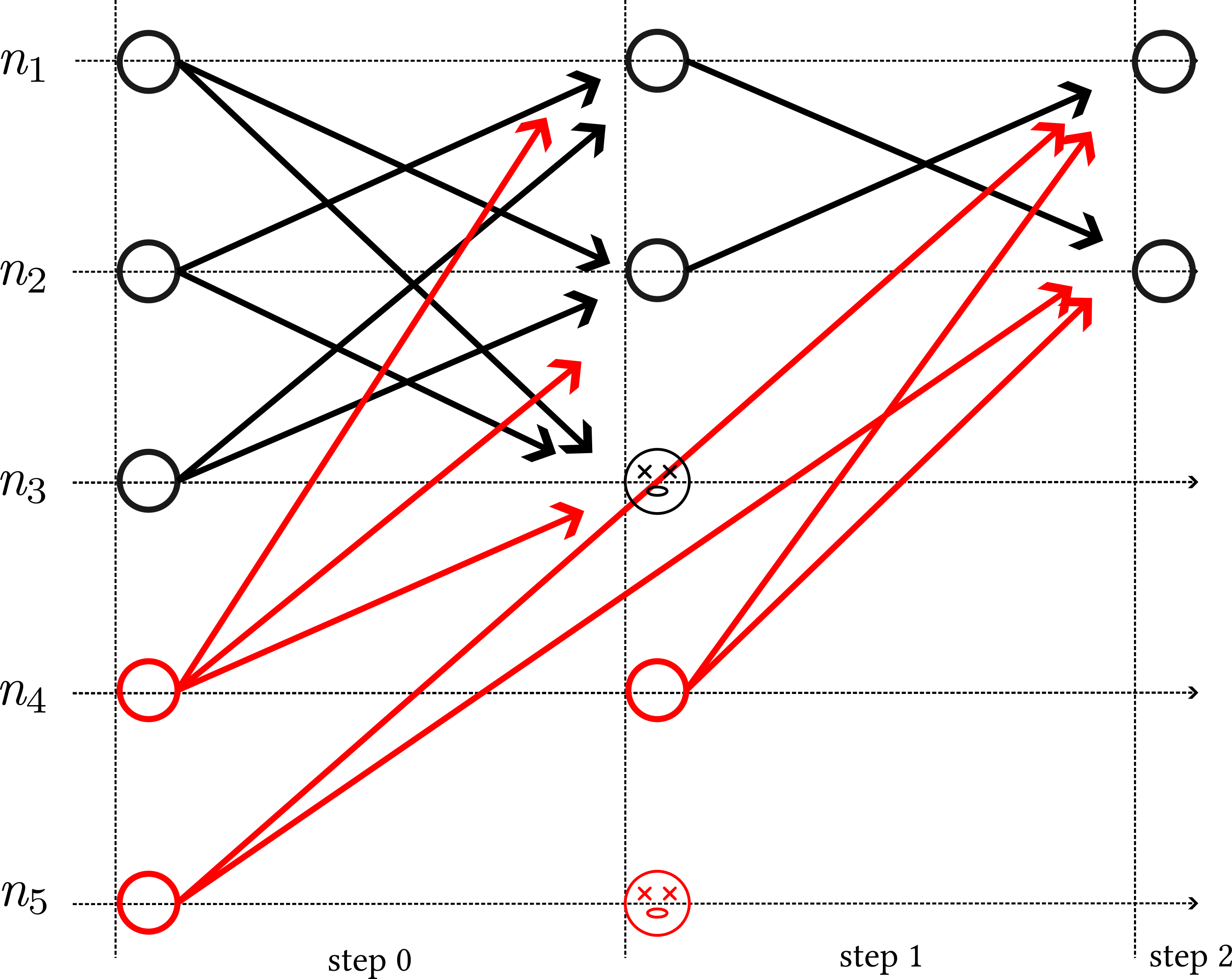}
    \caption{Example of time-travel attack. Although active correct nodes form a majority among active nodes in Step~1 ({\em i.e.}, 2 out of 3), correct nodes in Step~2 receive as many Byzantine messages as correct messages ({\em i.e.}, 2 out of 4).}%
    \label{fig:time-travel}
\end{figure}

\section{Model}\label{sec: model}

We consider a synchronous system equipped with a proof-of-work (PoW) primitive. The system is permissionless in that the set of participating nodes is unknown, and each node may become active or inactive at any time.

\textbf{Nodes.}
The system consists of an infinite set of nodes~$n_1, n_2, \dots$.
Each node is either \emph{correct} or \emph{Byzantine}.
Correct nodes follow their assigned protocol, while Byzantine nodes behave arbitrarily, subject to the PoW constraints described below.

\textbf{Ticks and steps.}
Time progresses in discrete \emph{ticks} numbered~$0, 1, 2, \dots$.
Every node has a clock indicating the current tick.
For some fixed integer parameter~$K \gg 1$, every~$K$ ticks are grouped into a \emph{step}; thus, each step~$i \geq 0$ consists of ticks~$\{iK, iK + 1, \dots, iK + K - 1\}$.

\textbf{Active and inactive nodes.}
At each tick, a non-zero, finite number of nodes are {\em active}; the rest are {\em inactive}. A node's status is determined by an unknown {\em activity function}.
Correct nodes change their status only at step boundaries---they are either active or inactive for an entire step.
We assume at least one correct node is active in every step.

\textbf{Computing power and the \dpow oracle.} The \dpow oracle produces \dpow evaluations and provides two methods.
The first allows a node to indicate the work it is willing to expend to obtain a \dpow; this value, in turn, determines the time it takes the oracle to respond, depending on the node's fixed \emph{computing power}~$\mathcal{P}(n)> 0$, which represents its hardware and energy budget.
The second method allows a node to verify whether a \dpow evaluation was produced by the oracle.

The \dpow oracle maintains a private map from pairs of the form~$\langle \gamma, w\rangle$, where~$\gamma$ is any value and~$w$ is a non-zero natural number called a \emph{weight}, to pairs of the form~$\langle dpow, s\rangle$, where~$dpow$ is a {\em unique} value called a \dpow evaluation and step~$s$ is~$\gamma$'s \emph{generation time}.
When~$\langle \gamma, w\rangle$ is mapped to~$\langle dpow,s \rangle$ for some~$s$, we say that~$dpow$ is a correct \dpow evaluation on~$\langle \gamma, w\rangle$.

Nodes can access the \dpow oracle via two methods:
\begin{itemize}
    \item $\textsc{dpow}(\gamma, w)$, where~$\gamma$ may be any value and~$w$ is a weight, representing the amount of work the caller wishes to expend.
          Upon a call~$\textsc{dpow}(\gamma, w)$ at a tick~$t$ by a node~$n$, the oracle ($i$) checks whether~$n$ has a pending \dpow evaluation.
          If so, it returns immediately without further action. Otherwise, the oracle ($ii$) checks whether it already has a mapping for~$\langle \gamma, w\rangle$.
          If not, it picks uniformly at random a \dpow evaluation~$dpow$ which does not appear in the map and \emph{registers} it by inserting the mapping~$\langle \gamma, w\rangle\rightarrow \langle dpow, s\rangle$, where~$s=\lfloor t/K \rfloor$ is the current step.
          Finally, ($iii$) the oracle schedules the delivery of the \dpow evaluation associated with~$\langle \gamma, w\rangle$ for the earliest tick~$t'> t$ such that ({\em a})~$n$ is active at~$t'$ and ({\em b})~$n$ has been active for a number of ticks~$\lceil wK/\mathcal{P}(n)\rceil$ between~$t$ (included) and~$t'$ (excluded).
         
          If a $\textsc{dpow}$~call and the corresponding delivery of its evaluation happen in steps~$s$ and~$s'$, respectively, we say that the call is \emph{within}~$[s, s']$.
          If~$s'=s$, we say the \dpow \emph{belongs} to step~$s$.
          When clear from context, we refer to a \dpow evaluation result as~$dpow$.
    \item $\textsc{verify}(dpow, \gamma, w)$, a Boolean function where~$dpow$ is a \dpow evaluation,~$\gamma$ is any value, and~$w$ is a weight.
          It returns true if and only if~$\langle \gamma, w\rangle$ appears in the oracle's internal map and is mapped to~$\langle dpow,s\rangle$ for some~$s$.
\end{itemize}
Adopting a Dolev-Yao~\cite{dolev-yao} approach, Byzantine nodes cannot obtain the \dpow evaluation of an input~$\langle\gamma, w\rangle$ without calling~$\textsc{dpow}(\gamma, w)$, unless they receive it in a message.

\textbf{The {\em correct supremacy} assumption.}
For a fixed real parameter~$0 \leq \rho\leq 1/2$, Byzantine nodes are~$\rho$-bounded:
For every interval of steps~$[s, s']$, let~$\Sigma_B$ be the sum of the weights of the \dpow evaluations of Byzantine nodes within the interval, and let~$\Sigma$ be the sum of the weights of \dpow evaluations of all nodes within the interval.
Then,~$\Sigma_B<\rho\Sigma$.

\textbf{Node and network behavior.} At each tick, every active correct node performs the following sequence of actions: ($i$) it receives a set of messages from the network, and possibly a \dpow evaluation scheduled for delivery at that tick by the \dpow oracle; ($ii$) it performs local computation, which may include invoking the \dpow oracle; and ($iii$) it broadcasts messages to the network.
Each message that a correct node disseminates is unique, {\em i.e.}, it has never appeared before in the system.
This uniqueness can be enforced, for example, by including random nonces in messages.
Byzantine nodes, when active, are not bound by the protocol. They may perform arbitrary computations, invoke the \dpow oracle with any arguments, and broadcast arbitrary messages.

The network is \emph{synchronous}, \emph{reliable}, and does not duplicate or generate messages. A message sent by a correct node at tick~$t$ is received at tick~$t+1$ by all correct nodes that are active at tick \(t+1\). That is, message delivery occurs atomically as the system transitions from tick~$t$ to tick~$t+1$.
Nodes that are inactive at tick~$t+1$ will receive these messages at the first tick~$t' > t$ at which they become active. We define a \emph{correct message} as one generated by a correct node; all other messages are considered \emph{Byzantine}.\footnote{Recall that correct nodes send unique, fresh messages, so there is no ambiguity in identifying them.}
Additionally, at each tick, every active correct node receives all Byzantine messages that have been observed in previous ticks by any correct node since it was last active.

\subsection{Total-Order Broadcast}
\label{sec:tob-def}

A \emph{transaction} is a string.
A \emph{block} is a data structure comprising a set of transactions and a pointer to another block.
A \emph{chain} is a collection~$\langle B_1,\dots, B_k\rangle$ of blocks where for every~$1 < i\leq k$, $B_i$ points to~$B_{i-1}$.

Each node consists of two components: a \emph{total-order broadcast} (TOB) {\em module} and a \emph{client module}. The client module can submit new blocks to the TOB module via a \texttt{submit} downcall. Conversely, the TOB module notifies the client of newly committed chains via a \texttt{commit} upcall.

The client continuously submits new blocks, so that each correct node always has at least one block that has been submitted but not yet included in any committed chain.

To formalize the behavior of a TOB implementation, we define the notions of \emph{compatible} and \emph{incompatible} chains. Given two chains $\Lambda_1$ and $\Lambda_2$, we say they are \emph{compatible} if one is a prefix of the other; otherwise, they are \emph{incompatible}.

\begin{definition}\label{def:tob}
    A TOB algorithm satisfies the following properties:
    \begin{itemize}
        \item Consistency: If two correct nodes commit chains $\Lambda_1$ and $\Lambda_2$, then $\Lambda_1$ and $\Lambda_2$ are compatible.
        \item Progress: Let $\Lambda$ be the longest chain committed by all correct nodes.
              At all times, with probability 1, $\Lambda$ eventually includes at least one more block submitted by a correct node.
    \end{itemize}
\end{definition}

\section{Sieve: Fending Off Time Travel}\label{sec: sieve}
Dynamically available protocols like that of Malkhi et al.~\cite{malkhi_towards_2023} use quorum intersection arguments.
These protocols count messages (carrying votes) and assume that each round contains a minimum fraction of correct messages.
In a permissionless model, however, this assumption exposes  them to \emph{time travel} attacks: even if correct nodes generate a majority of the messages during any time interval, Byzantine nodes can pre-generate messages and strategically release them later to outnumber correct messages.
We call such pre-prepared and delayed messages \emph{antique messages}.

To address this vulnerability, we introduce \emph{time-travel-resilient broadcast} (TTRB, \S\ref{sec: sieve-TTRB}), a broadcast abstraction designed to filter out antique messages. We then present Sieve, a protocol that implements TTRB using the \dpow oracle.
Sieve detects and prunes antique messages using two distinct filtering policies: \emph{Online-Sieve} and \emph{Bootstrap-Sieve}. 
Online-Sieve is efficient but can only be used by nodes that were already online in the previous step; nodes that newly become active must first catch up using Bootstrap-Sieve before they can switch to Online-Sieve.

We begin with an overview of Sieve's operation (\S\ref{sec: sieve-overview}), and proceed to a detailed description of the Sieve protocol followed by correct nodes (\S\ref{sec:sieve-main}) treating the filtering policies as black-box functions.
Finally, we describe the Online-Sieve~(\S\ref{sec:sieve-online}) and Bootstrap-Sieve~(\S\ref{sec:sieve-bootstrap}) mechanisms.

\subsection{Time-Travel-Resilient Broadcast}\label{sec: sieve-TTRB}

A time-travel-resilient broadcast protocol provides a black-box broadcast abstraction.
At each tick~$t$, the runtime invokes~$\textsc{UponNewTick}(t, \mathcal{M}, \mathcal{R})$, where~$\mathcal{M}$ is a set of messages received over the network and~$\mathcal{R}$ is a set of \dpow evaluations.
In response, TTRB may interact with the \dpow oracle (see Figure~\ref{fig:sieve-mmr-stack}).
When the tick~$t$ is the first tick of a step~$s$, TTRB makes a~$\textsc{TTRBDeliver}(s, L)$ up-call to the application, delivering a set of messages~$L$ of the form~$\langle m, v, w\rangle$.
For each tuple~$\langle m, v, w\rangle$,~$v$ is a DPoW evaluation with weight~$w$ and~$m$ is a message payload.
The application responds by returning  a message~$msg$ that  the TTRB layer then broadcasts by executing~$\textsc{TTRBCast}(msg)$.
TTRB guarantees the following:

\begin{definition}[TTRB implementation]
    \label{def:ttrb}
    For every correct node and every step~$s$, if the node calls~$\textsc{TTRBDeliver}(s, L)$ in step~$s$, then:
    \begin{enumerate}[label=TTRB\arabic*, ref=TTRB\arabic*, left=0pt]
        \item\label{TTRB1} For every tuple~$\langle m, v, w\rangle\in L$,~$v$~is a correct DPoW evaluation on~$m$ with weight~$w$ and oracle generation time~$s-1$.
        \item\label{TTRB2} For every correct node~$n$ active at step~$s-1$, if upon executing~$\textsc{TTRBDeliver}$ node~$n$ returns a message~$m$ to be~\textsc{TTRBcast} in step~$s-1$, then~$m$ appears in~$L$ with weight~$\mathcal{P}(n)$.
    \end{enumerate}
\end{definition}
A protocol is said to implement TTRB if it satisfies Definition~\ref{def:ttrb}.
When the step~$s$ is clear from the context, we also say that the set~$L$ from Definition~\ref{def:ttrb} satisfies TTRB.

\begin{figure}[ht]
    \centering
    \includegraphics[width=0.5\columnwidth]{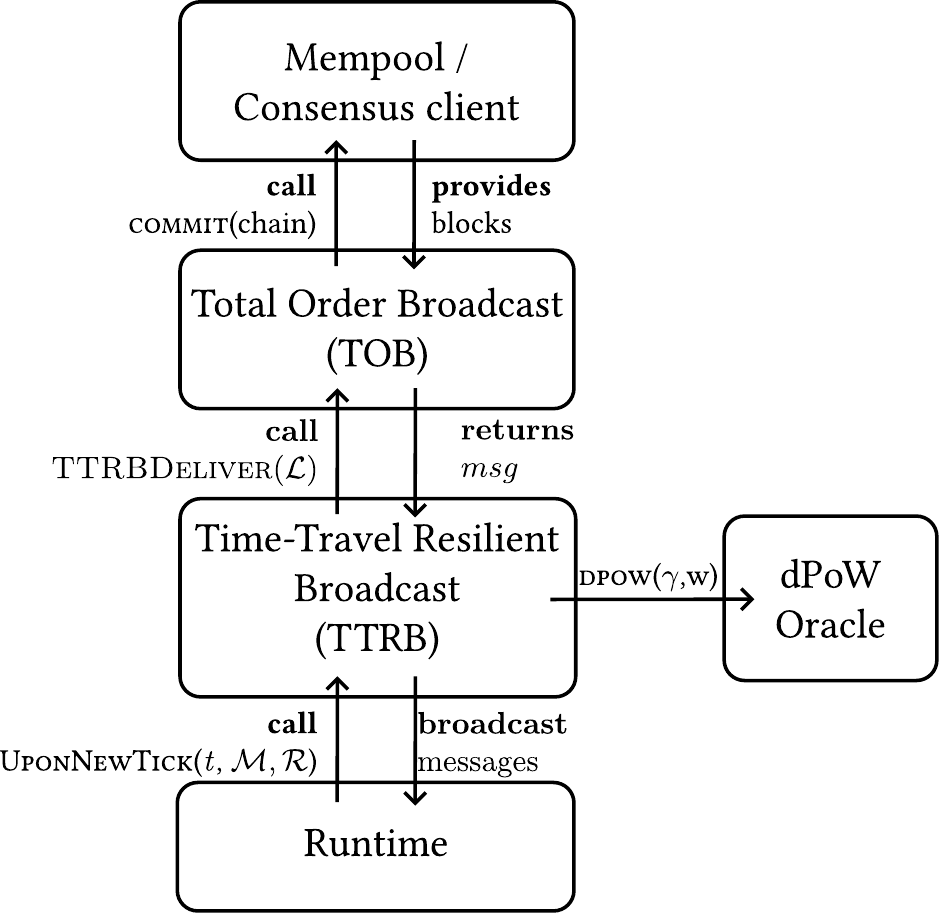}
    \caption{Protocol stack of the Sieve-MMR algorithm.}%
    \label{fig:sieve-mmr-stack}
\end{figure}

\subsection{Overview of Sieve}\label{sec: sieve-overview}

Sieve implements TTRB, assuming a~$\rho$-bounded adversary with~$\rho\leq 1/2$.
It does so by identifying antique messages and filtering them out from the set of messages received by a correct node at each step.

To filter out antique messages, Sieve leverages the synchronous nature of the model: it requires correct nodes to send messages only in the last tick of a step.
This means that a message generated by a correct node in some step is received by other nodes only in the first tick of the next step.
Now, consider a correct node at the first tick of some step~$s$ that has received a set of messages claiming to be from step~$s-1$ and wants to identify the antique ones among them.
An antique message~$m$ that claims to be from~$s-1$ when it is instead from some prior step~$s'$ cannot causally depend on a correct message~$m'$ sent in step~$s-2$, because~$m$ was generated before~$m'$ was sent ($m'$ was sent in the last tick of step~$s-2$).
Thus, Sieve discards antique messages by filtering out, at each step, the messages that are not causally preceded by a correct message from the previous step. However, identifying these causal relationships is challenging.

To make causal relationships between messages explicit and to preclude costless generation of messages, Sieve requires each node to include in every message
($i$)~a \emph{coffer} field containing messages from the previous step that the node deems non-antique and
($ii$)~the \dpow evaluation of the tuple consisting of the message payload, its coffer, and a random nonce (to make messages unique (\S\ref{sec: model})).
This gives rise to a DAG of \dpow evaluations (\S\ref{sec:sieve-main}), where vertices are messages and edges encode inclusion in message coffers.
Sieve analyzes the DAG by applying to it one of two filtering policies: Bootstrap-Sieve and Online-Sieve.

At each step~$s$, Bootstrap-Sieve computes a set of non-antique messages from step~$s-1$, which includes all correct ones, relying on all messages received thus far; Online-Sieve does the same, but relying only on messages claiming to belong to step~$s-1$ and a set of messages from step~$s-2$ satisfying TTRB (\S\ref{sec: sieve-TTRB}).
A correct node thus uses Bootstrap-Sieve at the first step it becomes active after a period of inactivity and then switches to Online-Sieve as long as it remains active.
If a node becomes inactive, it must run Bootstrap-Sieve again the next step that it becomes active.

\subsection{Sieve}\label{sec:sieve-main}

Sieve is detailed in~\Cref{alg:sieve}, where Bootstrap-Sieve and Online-Sieve are treated as black-box function calls.
Each node~$n$ maintains four state components: the set~$\mathcal{M}$ of messages it has received so far, the step~$\textit{last-active}$ in which it was last active, the set of non-antique messages~$\mathcal{L}$ received in that step, and a TTRB message~$\textit{pending-ttrb-msg}$ whose \dpow evaluation is currently pending.

Node~$n$ starts the execution by running its \textsc{Main} procedure (Line~\ref{algo:sieve:main}).
In this procedure, as long as~$n$ wants to be active, it calls the procedure~$\textsc{UponNewTick}(t, \mathcal{M}', \mathcal{R})$ (defined in Line~\ref{alg:sieve-upon-new-tick}) at each tick~$t$, where~$\mathcal{M}'$ is the set of messages that~$n$ has received since~$\textit{last-active}$, and~$\mathcal{R}$ is the set of oracle responses scheduled for~$t$. 
Node~$n$ first adds all messages in~$\mathcal{M}'$ to~$\mathcal{M}$ (Line~\ref{algo:sieve:add-to-received-msgs}).
Then, there are three cases, depending on whether the current tick is the first tick of the current step, the last tick of the current step, or neither:

\begin{itemize}
    \item If~$t$ is neither the first tick nor the last tick of the current step, nothing else happens.
    \item If~$t$ is the first tick of the current step~$s=\lfloor t/K\rfloor$ (so~${t\bmod K=0}$), then~$n$ calls~$\textsc{NewStep}(s)$ (Line \ref{algo:sieve:new-step-call}).
          In~$\textsc{NewStep}(s)$, depending on whether~$n$ was active in the last step or not,~$n$ determines the set $\mathcal{L}$ of non-antique messages from the previous step~$s-1$ using either Online-Sieve (Line \ref{algo:sieve:call-online-sieve}) or Bootstrap-Sieve (Line \ref{algo:sieve:call-bootstrap-sieve}), and records this set in~$\mathcal{L}$.
          Then Sieve delivers~$s$ and~$\mathcal{L}$ to the application (Line \ref{algo:sieve:call-deliver}).

          The application returns a message for broadcast, which is assigned to~$\textit{ttrb-msg}$.
          Node~$n$ then executes~$\textsc{TTRBCast}(\textit{ttrb-msg})$ (defined in Line \ref{algo:sieve:ttrbcast}), forming the triple~$\gamma=\left\langle \textit{ttrb-msg}, \mathcal{L}, r \right\rangle$, where~$r$ is a random nonce.
          It then records \textit{ttrb-msg} as the current message with a pending \dpow evaluation by assigning it to the variable \textit{pending-ttrb-msg},
          and calls the \dpow oracle (Line~\ref{algo:sieve:call-dpow}) to obtain a \dpow evaluation of~$\gamma$ with weight~\(\mathcal{P}(n)\).
          Note that the weight is chosen so that~$n$ will receive the \dpow response in the last tick of the current step, and thus~$n$ will be able to broadcast it to all correct nodes by the end of the step.
    \item
          If~$t$ is the last tick of the current step (so~$t\bmod K=K-1$), since~$n$ has called the \dpow oracle at the first tick of the step with weight~$\mathcal{P}(n)$,~$n$ receives a set of \dpow responses~$\mathcal{R}=\{dpow\}$ where~$dpow$ is the \dpow evaluation corresponding to the pending TTRB message in the variable \textit{pending-ttrb-msg}.
          Node~$n$ then broadcasts on the network the message~$m=\langle\textit{pending-ttrb-msg}, \lfloor t/K\rfloor, \mathcal{L}, dpow, \mathcal{P}(n)\rangle$, where~$\lfloor t/K\rfloor $ is the current step,~$\mathcal{L}$ is the set of non-antique messages computed during the first tick of the step,~$dpow$ is the \dpow evaluation just received, and~\(\mathcal{P}(n)\) is the weight of~$dpow$.
\end{itemize}

\begin{algorithm}
    \caption{The Sieve algorithm, code for node $n$.}%
    \label{alg:sieve}
    \begin{algorithmic}[1]
        \LeftComment{State variables:}
        \State $\mathcal{M} \gets \{\}$ \Comment{Set of messages received so far}
        \State{$\textit{last-active} \gets -1$}\Comment{Last step in which the node was active}
        \State{$\mathcal{L} \gets \{\}$}\Comment{Non-antique messages received in last active step}
        \State{$\textit{pending-ttrb-msg}\gets \text{none}$}\Comment{TTRB message waiting for \dpow evaluation}
        \Procedure{UponNewTick}{$t$, $\mathcal{M}'$, $\mathcal{R}$}\label{alg:sieve-upon-new-tick}
            \State $\mathcal{M}\gets \mathcal{M}\cup \mathcal{M}'$\label{algo:sieve:add-to-received-msgs}
            \If{$t\bmod K = 0$}\Comment{First tick of the step}
                \State $\Call{NewStep}{t/K}$\label{algo:sieve:new-step-call}
            \ElsIf{$t \bmod K = K-1$}\Comment{Last tick of the step}
                \State{$\{dpow\}\gets \mathcal{R}$} \Comment{extracts the single \dpow response}
                \State{$m \gets \langle\textit{pending-ttrb-msg}, \lfloor t/K\rfloor, \mathcal{L}, dpow, \mathcal{P}(n)\rangle$}
                \State \Call{Broadcast}{$m$} \Comment{Broadcast on the network}
            \EndIf
        \EndProcedure
        \Procedure{NewStep}{$s$}
            \If{$\textit{last-active} = s-1$}
                \State{$\mathcal{L} \gets$ \Call{OnlineSieve}{$s,\mathcal{M},\mathcal{L}$}}\label{algo:sieve:call-online-sieve}
            \Else
                \State{$\mathcal{L} \gets$ \Call{BootstrapSieve}{$s,\mathcal{M}$}}\label{algo:sieve:call-bootstrap-sieve}
            \EndIf
            \State{$\textit{last-active} \gets s$}
            \State{$\textit{ttrb-msg} \gets \textsc{TTRBDeliver}(s, \mathcal{L}$})\Comment{Call the upper layer}\label{algo:sieve:call-deliver}
            \State{$\textsc{TTRBCast}(\textit{ttrb-msg}$)}\label{alg:sieve-self-ttrbcast-call}
        \EndProcedure
        \Procedure{TTRBCast}{\textit{ttrb-msg}}\label{algo:sieve:ttrbcast}
            \State{\textit{r} $\gets$ a random number}
            \State{$\gamma \gets \langle \textit{ttrb-msg}, \mathcal{L}, r\rangle$}
            \State{$w \gets \mathcal{P}(n)$}\Comment{Weight of the \dpow}
            \State{$\textit{pending-ttrb-msg}\gets \textit{ttrb-msg}$}
            \State{\Call{dpow}{$\gamma$, $w$}\Comment{Call the \dpow oracle}}\label{algo:sieve:call-dpow}
        \EndProcedure
        \Procedure{Main}{}\label{algo:sieve:main}\Comment{Execution starts from here}
        \While{$n$ wants to be active}
        \State{$t \gets \textsc{time}.\textit{now}()$}
        \State{$\mathcal{M}' \gets \text{messages received since}\ \textit{last-active}$}
        \State{$\mathcal{R} \gets \text{oracle responses received for tick}\ t$}
        \State{$\textsc{UponNewTick}(t,\mathcal{M}',\mathcal{R})$}
        \State{$\textsc{Wait For tick}\ t+1$}
        \EndWhile
        \EndProcedure
    \end{algorithmic}
\end{algorithm}

Before moving on to explain Online-Sieve and Bootstrap-Sieve, we need the following concepts.

\textbf{Step, weight, and timestamp of messages.}
Given a Sieve message~$m=\langle \textit{ttrb-msg}, s, \mathcal{L}, dpow, w\rangle$, we say that~$m$ is a \emph{timestamp-$s$ message} with weight~$w$, and we also write~$\texttt{weight}(m)$ for~$w$.
Given a set~$M$ of messages, we write~$M_s$ for the set of timestamp-$s$ messages in~$M$.
Moreover, abusing notation, the weight~$\texttt{weight}(M)$ of~$M$ is~$\sum_{m\in M}\texttt{weight}(m)$.
Given two sets of messages~$M_1\subseteq M_2$ and~$0 < \rho \leq 1$,~$M_1$ is \emph{strictly more than a weighted fraction~$1-\rho$ of}~$M_2$ if~$\texttt{weight}(M_1) > (1-\rho)\cdot\texttt{weight}(M_2)$.

We say a message~$m$ is \emph{generated at step~$s$} if~$s$ is the generation time associated in the \dpow oracle with the \dpow that~$m$ carries.
As guaranteed in our model by the \dpow oracle, this is the step at which the \dpow oracle was called to obtain the \dpow evaluation attached to the message.
We use it in our definitions and proofs, but it is not accessible to the nodes.
We sometimes refer to the coffer of a message~$m$ as~$\operatorname{coffer}(m)$.

Note that a timestamp-$s$ message is not necessarily generated at step~$s$, as an antique message can maliciously claim that it belongs to step~$s$.
Similarly, Byzantine nodes can assign any weight they want to a message.

\textbf{Message DAGs.}
A set of messages~$M$ forms a directed acyclic graph (DAG) defined as the graph whose vertices are the messages in~$M$, such that there is an edge from message~$m_1$ to message~$m_2$ if and only if~$m_2$ is in~$m_1$'s coffer. No cycles are possible because the \dpow oracle is a random oracle, and thus each new \dpow evaluation is a fresh random value. Depending on the context, we will refer to a collection of messages as both a set and a DAG.

\subsection{Online-Sieve}\label{sec:sieve-online}

Consider a correct node~$n$ at some step~$s$ that is scrutinizing a correct timestamp-$(s-1)$ message $m$ to check whether it is antique, and assume that~$n$ has a set~$\mathcal{L}$ of timestamp-$(s-2)$ messages satisfying TTRB.
Moreover, assume inductively that the execution up to step~$s-1 \geq 0$ satisfies TTRB.

\subsubsection{Intuition}
Both~$\mathcal{L}$ and~$m$'s coffer satisfy TTRB.
Consider~$\mathcal{L}$: (i) it contains all the correct timestamp-$(s-2)$ messages, and (ii) all messages in~$\mathcal{L}$ are generated at~$s-2$.
This means, according to our correct supremacy assumption, that the set of correct timestamp-$(s-2)$ messages~$\mathcal{C}\subseteq \mathcal{L}$ is strictly more than a weighted fraction~$1-\rho$ of~$\mathcal{L}$.
The same logic works for~$m$'s coffer as well.

This observation gives us a filtering rubric that can be checked efficiently.
If a timestamp-$(s-1)$ message~$m$ is correct, the intersection of~$\mathcal{L}$ and~$m$'s coffer should be strictly more than a weighted fraction~$1-\rho$ of~$\mathcal{L}$, as it contains all correct timestamp-$(s-2)$ messages.
If, on the other hand,~$m$ is antique, its sender must have sent it before receiving any of the correct messages in~$\mathcal{L}$, because they did not exist yet.
Thus the intersection of its coffer with $\mathcal{L}$ will consist of less than a weighted fraction $1-\rho$ of $\mathcal{L}$.

\subsubsection{Algorithm}

The \textsc{OnlineSieve} sub-procedure receives three arguments: the current step~$s$, the set~$\mathcal{M}$ of messages received in~$s$, and the set~$\mathcal{L}$ of timestamp-$(s-2)$ non-antique messages computed in the previous step~$s-1$ (except~$\mathcal{L}=\emptyset$ if~$s<2$).
It must return a subset of~$\mathcal{M}$ containing all correct timestamp-$(s-1)$ messages and excluding any antique timestamp-$(s-1)$ message.

\Cref{alg:online-sieve} presents a pseudocode description of Online-Sieve.
First, Online-Sieve filters out from~$\mathcal{M}$ all messages that are not timestamp-$(s-1)$ messages or that have an invalid \dpow evaluation.
Then, out of the remaining messages, Online-Sieve selects every message~$m$ such that the weight of the messages in common between~$m$'s coffer and~$\mathcal{L}$ is more than a weighted fraction~$1-\rho$ of~$\mathcal{L}$.

\begin{algorithm}[t]
    \caption{Online-Sieve. Using the set~$\mathcal{L}$ of non-antique messages computed in the previous step~$s-1$, Online-Sieve filters out antique messages from the set~$\mathcal{M}$ of timestamp-$(s-1)$ messages received in step~$s$.}%
    \label{alg:online-sieve}
    \begin{algorithmic}[1]
        \Procedure{OnlineSieve}{$s, \mathcal{M}, \mathcal{L}$}\label{alg:opt-inputs}
            \State{$\mathcal{M}_{s-1} \gets \{m \in \mathcal{M} \mid \text{$m$ has timestamp~$s-1$}\}$}
            \State{$\mathcal{V}_{s-1} \gets \{m \in \mathcal{M}_{s-1} \mid \text{$m$ has a valid \dpow}\}$}
            \State{\Return$ \{m\in \mathcal{V}_{s-1} \mid \texttt{weight}(\text{coffer}(m)\cap \mathcal{L}) > (1-\rho)\cdot\texttt{weight}(\mathcal{L})\}$}\label{alg:opt-forall-m}\label{alg:opt-ls-rm}
        \EndProcedure
    \end{algorithmic}
\end{algorithm}

\subsubsection{Example}%
\label{sec:online-sieve-example}

\begin{figure}[h]
    \centering
    \includegraphics[width=.5\columnwidth]{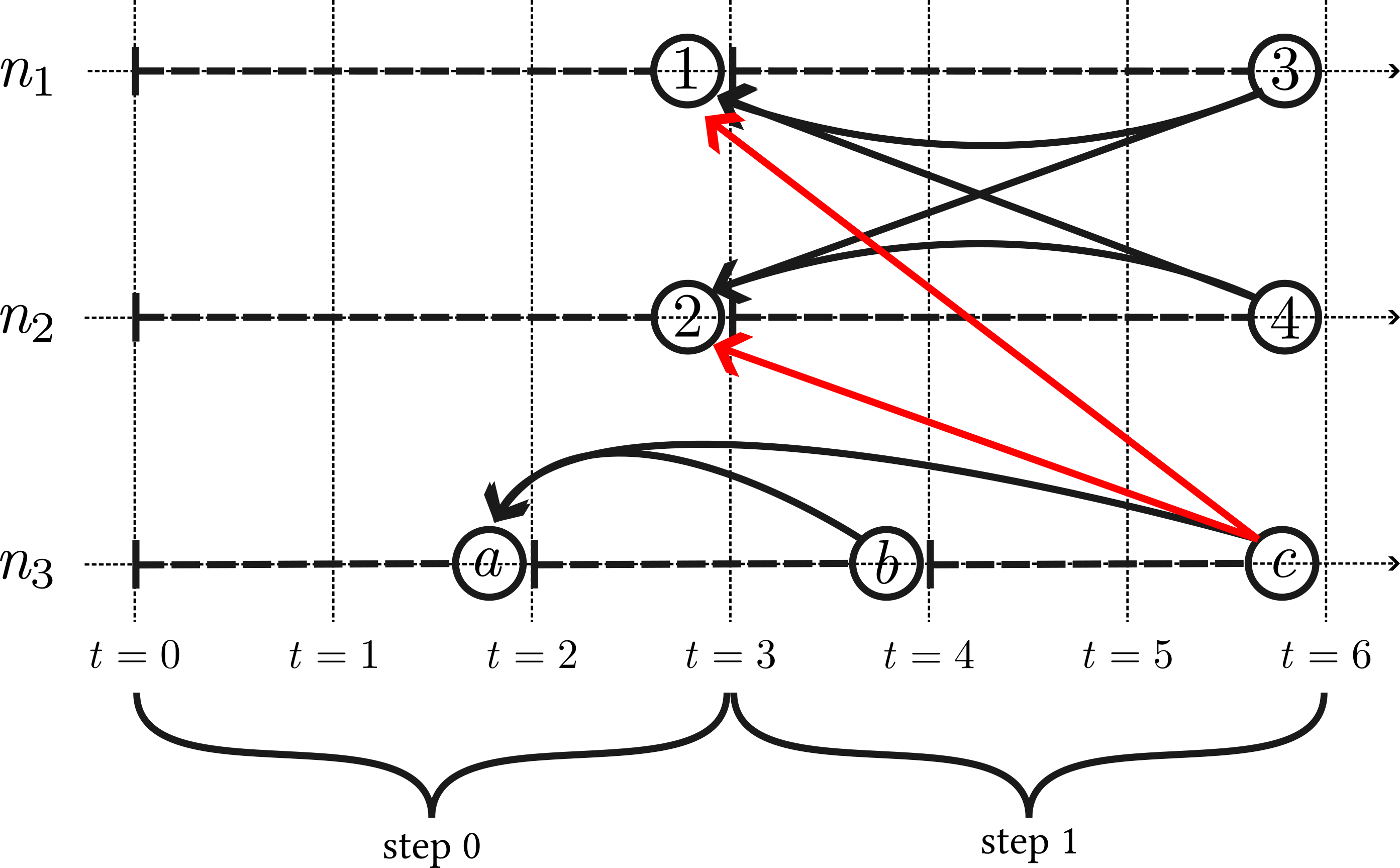}
    \caption{Example execution in which correct nodes~$n_1$ and~$n_2$ use Online-Sieve and a Byzantine node~$n_3$ produces an antique message \msg{b}. When building their timestamp-$2$ messages at~$t=6$, correct nodes must filter out \msg{b}.}%
    \label{fig:online-sieve-example}
\end{figure}

Consider the execution depicted in Figure~\ref{fig:online-sieve-example}, where~$K=3$.
Nodes~$n_1$ and~$n_2$ are correct and have computing power~$1$, which means they obtain \dpow evaluations of weight 1 every 3 ticks.
Node~$n_{3}$ is Byzantine with computing power~$1.5$, and obtains \dpow evaluations of weight 1 every 2 ticks.
Note that correct supremacy holds: in any interval of steps~$[s,s']$, the ratio of Byzantine messages is at most~$3/7$, which is smaller than~$1/2$.

Dashed horizontal segments represent time intervals during which nodes are waiting for \dpow responses, ending with a circle that represents a message and its \dpow response; the edges connecting messages represent coffer inclusion.
Red edges signify a coffer inclusion relation between correct and Byzantine messages.
For example, messages~\msg{3} and~\msg{4} hold messages~\msg{1} and~\msg{2} in their coffers, while message~\msg{c}'s coffer contains the correct messages~\msg{1},~\msg{2}, and the Byzantine message~\msg{a}.
Messages~\msg{1},~\msg{2}, and~\msg{a} are timestamp-$0$ messages, while~\msg{3},~\msg{4},~\msg{b}, and~\msg{c} are timestamp-$1$ messages.
Note that~\msg{b} is antique, since it has timestamp~$1$ but it started before step~$1$,
and that its coffer does not and cannot possibly contain any of the correct timestamp-0 messages because they had not been generated yet when~\msg{b} started.

Consider~$n_1$ at step~$2$ ($t=6$), applying Online-Sieve to~$\{\msg{3}, \msg{4}, \msg{b}, \msg{c}\}$ to eliminate antique messages.
Online-Sieve does not filter any timestamp-0 messages, so~$n_1$ has~$\mathcal{L} = \{\msg{1}, \msg{2}, \msg{a}\}$.
Let us scrutinize \msg{3}; the case for \msg{4} and \msg{c} is similar.
The coffer of \msg{3} is~$\{\msg{1}, \msg{2}\}$, so its intersection with~$\mathcal{L}$ is~$\{\msg{1}, \msg{2}\}$; the intersection accounts for~$2/3$ of~$\mathcal{L}$'s weight, and \msg{3} will not be discarded.
Now consider \msg{b}: its coffer is~$\{\msg{a}\}$, which has an intersection~$\{\msg{a}\}$ with~$\mathcal{L}$.
The intersection accounts for less than~$1/2$ of~$\mathcal{L}$'s weight, and \msg{b} will be discarded.
We conclude that~$n_1$ obtains the set~$\{\msg{3}, \msg{4}, \msg{c}\}$ via Online-Sieve, which satisfies TTRB: ($i$) they are all timestamp-1 messages, and ($ii$) correct timestamp-1 messages are strictly more than a weighted fraction~$1/2$ of the set.

\subsection{Bootstrap-Sieve}\label{sec:sieve-bootstrap}

\begin{algorithm}[t]
    \caption{Bootstrap-Sieve on an input set of messages~$\mathcal{M}$ for a correct node $n$ at step~$s$.}%
    \label{alg:bootstrap-sieve}
    \begin{algorithmic}[1]
        \Procedure{BootstrapSieve}{$s, \mathcal{M}$}
        \State{$\tilde{\mathcal{L}} \gets \operatorname{VerifyDPoWsRecursively}(\mathcal{M})$}
        \label{alg-create-dag}
        \For{$s'\gets 1 \dots s-1$}\label{alg:bootstrap-start-iterative-pruning}
        \Comment{Removal Phase}
            \ForAll{$m \in \tilde{\mathcal{L}}_{s'}$}\label{alg:bootstrap-pick-messages}
            \Comment{The outcome of Removal Phase might depend on the order.} %
                \State{$C \gets$ a timestamp-$(s'-1)$ consistent DAG within~$\tilde{\mathcal{L}}$ containing~$m$, with maximal weight}
                \label{alg:bootstrap-find-consistent-dag}
                \If{no such~$C$ exists}
                    \State{$\tilde{\mathcal{L}} \gets \tilde{\mathcal{L}} \setminus \{m\}$}
                    \label{alg:bootstrap-remove-no-dag}
                \Else
                \State{$A \gets$ the seed of~$C$} %
                \If{$\exists B \subseteq \tilde{\mathcal{L}}_{s'-1} : A \cap B = \emptyset$ and~$C$ has a lower weight than a DAG consistent with~$B$}

                    \State{$\tilde{\mathcal{L}} \gets \tilde{\mathcal{L}} \setminus \{m\}$}
                    \label{alg:bootstrap-remove-message-heavier-dag}
                \EndIf
                \EndIf
            \EndFor
        \EndFor

        \Return{$\tilde{\mathcal{L}}_{s-1}$}
        \EndProcedure
    \end{algorithmic}
\end{algorithm}

Online-Sieve relies on an up-to-date set~$\mathcal{L}$ of non-antique messages computed in the previous step.
If the node was not active in the previous step, it has no such set~$\mathcal{L}$ available and it therefore cannot use Online-Sieve.
This is where Bootstrap-Sieve enters the picture: it allows a node newly active in a step~$s$ to compute, based on messages it has received so far, a set~$\mathcal{L}$ of timestamp-$(s-1)$ messages containing all correct timestamp-$(s-1)$ messages and no antique messages.

\subsubsection{Intuition}

Consider a node~$n$ newly active in step~$s>1$.
As per the model, in step~$s$,~$n$ receives a set of messages including all the messages sent by correct nodes in all steps~$s'<s$.
To implement TTRB, node~$n$ must now filter out all antique timestamp-$(s-1)$ messages.
To do this, node~$n$ iteratively filters out antique timestamp-$s'$ messages for~$0<s'<s$.
Each iteration~$s'$ relies on having filtered out the antique messages in all steps before~$s'$ (this is trivially the case if~$s'=1$ since by definition there cannot be antique timestamp-0 messages).

Next we informally explain how node~$n$ filters out antique timestamp-$s'$ messages assuming it has already filtered out all antique messages from previous steps.
The idea relies on the notions of consistent successors of a set of messages and of consistent DAGs of messages, which we define next.
\begin{definition}[Consistent successor]\label{def:consistent-succ}
    Given a set of messages~\,$X$, a message~$m$ is a consistent successor of \,$X$ when \,$X$ is a subset of~$m$'s coffer and \,$X$ is strictly more than a weighted fraction~$1-\rho$ of~$m$'s coffer.
\end{definition}
\begin{definition}[DAG of messages consistent with a set of timestamp-$s$ messages]\label{def:consistent-dag}
    A set of messages~$C$ is a consistent DAG when~$C$ is of the form~$C = X_{s}\cup X_{s+1}\cup X_{s+2}\cup \dots$, where for every integer~$s'\ge s$, every member of~$X_{s'+1}$ is a consistent successor of~$X_{s'}$.
    When all messages in~$X_{s}$ are timestamp-$s$ messages, we say that~$C$ is a timestamp-$s$ consistent DAG, or just a timestamp-$s$ DAG when clear from the context.
    We also say that~$C$ is a DAG consistent with~$X_{s}$, and that~$X_{s}$ is the seed of~$C$.
\end{definition}
Let us call the set of correct messages sent in step~$s'-1$ or later as~$C_{s'-1}^+$.
Note that, by the definition of TTRB, in every execution satisfying TTRB,~$C_{s'-1}^+$ forms a consistent DAG.

Now consider an antique timestamp-$s'$ message~$m$.
Message~$m$'s generation time is before~$s'$, which means it is also before any correct timestamp-$(s'-1)$ messages were sent.
Thus,~$m$'s coffer does not contain any correct timestamp-$(s'-1)$ messages.
Therefore, if~$B_{s'-1}$ is a timestamp-$(s'-1)$ consistent DAG containing the coffer of~$m$, then~$B_{s'-1}$ must be disjoint from~$C_{s'-1}^+$; otherwise, some correct message's coffer would contain both a supermajority of correct messages and a supermajority of Byzantine messages, which is not possible.
We prove this formally in~\Cref{lem:disjoint-remains-disjoint}.

Finally, since we have assumed that node~$n$ has already eliminated all antique timestamp-$(s'-1)$ messages, we have that all messages in both~$B_{s'-1}$ and~$C_{s'-1}^+$ were generated in step~$s'-1$ or after.
Thus, by the correct supremacy assumption,~$B_{s'-1}$ has strictly lower weight than~$C_{s'-1}^+$.
The idea is then, for each message~$m$, to ($i$) look for some heaviest consistent DAG containing~$m$, and to ($ii$) discard~$m$ if there exists a disjoint and heavier consistent DAG.

\subsubsection{Algorithm}

A pseudocode description of Bootstrap-Sieve appears in~\Cref{alg:bootstrap-sieve}.
Bootstrap-Sieve takes the current step~$s$ and the set of messages received so far,~$\mathcal{M}$, as input.
Then, for each message~$m$ in~$\mathcal{M}$, the node~$n$ verifies, using the \dpow oracle, that all the \dpow{}s of all the messages reachable from~$m$ in the DAG~$\mathcal{M}$ are valid.
Any message with an invalid \dpow is eliminated, and the remaining set of messages is assigned to the variable~$\tilde{\mathcal{L}}$.
The node then starts the iterative pruning process (Line~\ref{alg:bootstrap-start-iterative-pruning}).
At each iteration~$s'$ and for each timestamp-$s'$ message~$m$ it first identifies a heaviest timestamp-$(s'-1)$ DAG~$C$ containing~$m$ and consistent with some seed within~$\tilde{\mathcal{L}}$ (Line~\ref{alg:bootstrap-find-consistent-dag}), and rejects~$m$ if ($i$) no such consistent DAG exists (Line~\ref{alg:bootstrap-remove-no-dag}) or if ($ii$) there exists a heavier timestamp-$(s'-1)$ DAG consistent with some other seed that is also disjoint from~$C$ (Line~\ref{alg:bootstrap-remove-message-heavier-dag}).
It finally returns the set of timestamp-$(s-1)$ messages remaining in~$\tilde{\mathcal{L}}$.

Note that the algorithm might produce different results depending on the order in which messages are selected in Line~\ref{alg:bootstrap-pick-messages} of Algorithm~\ref{alg:bootstrap-sieve}.
Specifically, while Bootstrap-Sieve provably discards antique messages, it gives no guarantees on other Byzantine messages.
It might discard or keep Byzantine messages that do contain at least one correct message from the corresponding previous step, depending on the order in which messages are selected for scrutiny.

\subsubsection{Example 1: Online-Sieve is Not Enough}
Consider a node~$n$ that joins the execution in Figure~\ref{fig:bootstrap-sieve-example} at step 2, where objects have the same semantics as in Figure~\ref{fig:online-sieve-example}.
We show that iteratively running Online-Sieve leads to a violation of TTRB, and we need Bootstrap-Sieve.
The Byzantine node~$n_3$ poses messages~$\{\msg{a}, \msg{b}\}$ and~$\{\msg{c}\}$ as timestamp-0 and timestamp-1 messages, respectively.
Let~$\mathcal{L}_0$ and~$\mathcal{L}_1$ be the set of timestamp-0 and timestamp-1 messages that~$n$ obtains after removing antique messages, respectively.
Sieve does not discard timestamp-0 messages, so~$n$ will obtain~$\mathcal{L}_0 = \{\msg{1}, \msg{2}, \msg{a}, \msg{b}\}$.
Consider now the timestamp-1 message \msg{3}, whose coffer is~$\{\msg{1}, \msg{2}\}$.
The intersection of this coffer with~$\mathcal{L}_0$ is~$\{\msg{1}, \msg{2}\}$, which is not strictly more than a weighted fraction~$1/2$ of~$\mathcal{L}_0$.
Node~$n$ thus discards, in direct violation of TTRB, the correct timestamp-1 message \msg{3}.
This violation arises from the fact that~$\mathcal{L}_0$, \emph{at the time that~$n$ derives it}, does not satisfy TTRB: it has an equal number of correct and Byzantine messages.
A correct node that was present during the entire execution would have~$\mathcal{L}_0 = \{\msg{1}, \msg{2}, \msg{a}\}$, and thus would not have discarded message \msg{3}.
Note that our correct supremacy assumption is intact: at step 0, message \msg{b} was still not around and correct timestamp-0 messages were a majority.
It is only later on that Byzantine nodes can use their computational power, represented by \msg{b}, to confuse a newly joining correct node at step 2 when it is trying to reconstruct the history of the execution.
Also note that there are no antique messages here; this attack simply shows that Online-Sieve on its own is vulnerable even without time travel attacks.

\begin{figure}[ht]
    \centering
    \includegraphics[width=.5\columnwidth]{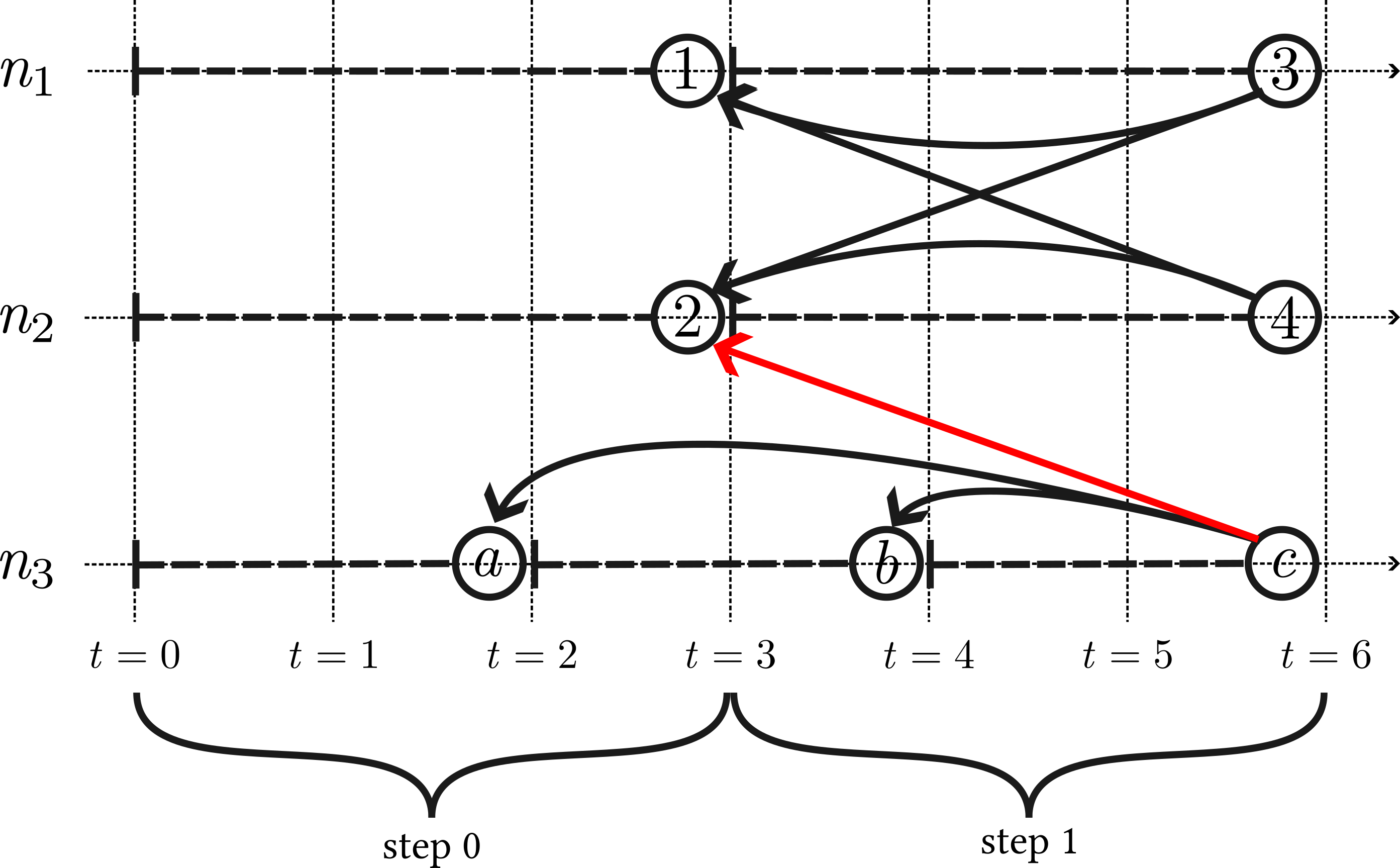}
    \caption{An execution showing that Online-Sieve on its own is not enough.
      It shows what a new node~$n$ joining at step 2 sees;~$n$ has to identify antique messages---there are none---but applying Online-Sieve at step 1 and then at step 2 ends up discarding correct timestamp-1 messages.}%
    \label{fig:bootstrap-sieve-example}
\end{figure}

\subsubsection{Example 2: Bootstrap-Sieve Locates Antique Messages}
Consider Figure~\ref{fig:bootstrap-sieve-prunes-antique}, where nodes~$n_1$ and~$n_2$ are correct,~$n_3$ is Byzantine, and message \msg{b} is an antique message claiming to belong to step 1.
Consider a correct node that at step~$2$ receives timestamp-1 messages \msg{3}, \msg{4}, \msg{b}, and \msg{c}.
We show that Bootstrap-Sieve ($i$) retains correct messages \msg{3} and \msg{4}, and ($ii$) discards \msg{b}.

\begin{figure}[ht]
    \centering
    \includegraphics[width=.5\columnwidth]{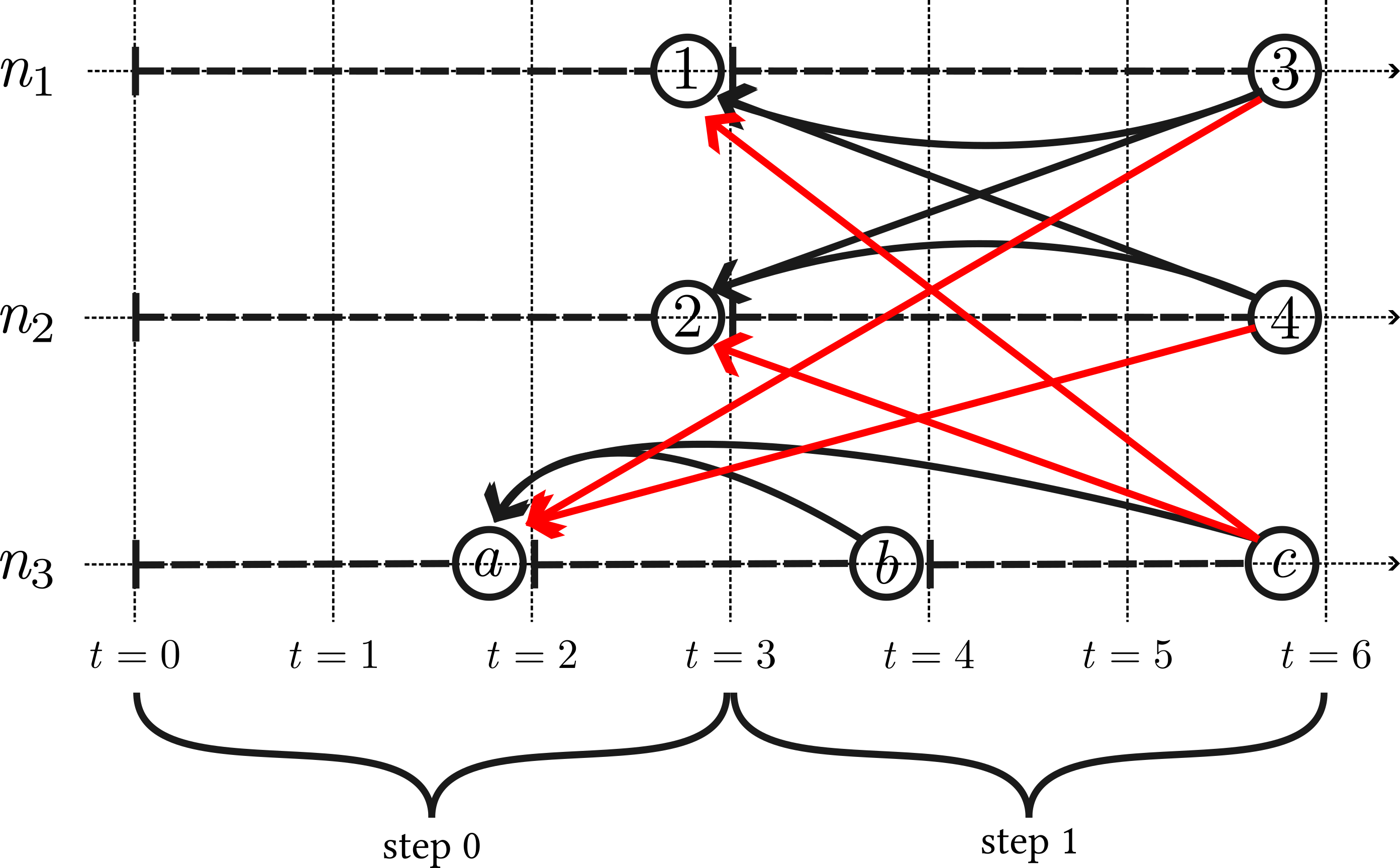}
    \caption{An example execution showing Bootstrap-Sieve in action.
    A newly joining correct node~$n$ at step 2 should identify \msg{b} as an antique timestamp-1 message.}%
    \label{fig:bootstrap-sieve-prunes-antique}
\end{figure}

Correct messages \msg{3} and \msg{4} have identical coffers, so Bootstrap-Sieve treats them the same; we are then only going to focus on
message \msg{3}.
Since \msg{3} is a timestamp-1 message, to determine its fate we need to identify the heaviest DAG -- call it~$\mathcal{C}$ -- that includes \msg{3} and is consistent with a subset of timestamp-0 messages ({\em i.e.}, with some subset of messages \msg{1}, \msg{2}, and \msg{a}).

In Figure~\ref{fig:bootstrap-sieve-prunes-antique},~$\mathcal{C}$
comprises vertices \{\msg{1},\msg{2}, \msg{a},\msg{3}, \msg{4}, \msg{c}\}.
Note that~$C$ is timestamp-0 consistent, as~$\{\msg{1}, \msg{2}, \msg{a}\}$ is a majority set in  the coffers of all messages in~$\{\msg{3}, \msg{4}, \msg{c}\}$, and thus \msg{3}, \msg{4}, and \msg{c} are all consistent successors of~$\{\msg{1}, \msg{2}, \msg{a}\}$.
There are no heavier timestamp-0 consistent DAGs disjoint from~$\mathcal{C}$; thus, \msg{3} is retained.

Consider now message \msg{b}.
The heaviest consistent DAG containing \msg{b} -- call it~$\mathcal{A}$ -- consists of vertices \msg{a} and \msg{b}. This time,
there exists a consistent DAG disjoint from and heavier than~$\mathcal{A}$, {\em i.e.},  the DAG with vertices \{\msg{1}, \msg{2}, \msg{3}, \msg{4}\}.
The set~$\{\msg{1}, \msg{2}\}$ is a majority in the coffers of \msg{3} and \msg{4}, which makes \{\msg{1}, \msg{2}, \msg{3}, \msg{4}\} a timestamp-0 consistent DAG.
Bootstrap-Sieve thus discards \msg{b}.

Note that Bootstrap-Sieve does not discard Byzantine message~\msg{c}.
This is how it should be, because~\msg{c} is not antique: indeed, Bootstrap-Sieve cannot distinguish it from a correct message! To see why, note that
the heaviest timestamp-0 consistent DAG containing~\msg{c}, ({\em i.e.}, \{\msg{1}, \msg{2}, \msg{a}, \msg{3}, \msg{4}, \msg{c}\}), is identical to the heaviest such DAG for correct messages~\msg{3} and~\msg{4}.
Thus~\msg{c}, like~\msg{3} and~\msg{4}, is retained.

\subsection{Practical Considerations}

Our model (\S\ref{sec: model}) is idealized and does not capture all aspects of real-world deployments.
There are important points that actual implementations of Sieve have to consider.

\textbf{Atomicity of actions.}
We have assumed in Algorithm~\ref{alg:sieve} that Bootstrap-Sieve and Online-Sieve execute atomically within a tick.
While this is plausible for Online-Sieve, which processes only messages from the previous step, it is not for Bootstrap-Sieve.
Execution history grows linearly, and the wall-clock time of executing Bootstrap-Sieve can get arbitrarily long.
Suppose that a node starts executing Bootstrap-Sieve at a step~$s$ and has not finished by the end of step~$s$.
The node must then buffer new messages it receives in steps greater than~$s$.
When it finishes executing Bootstrap-Sieve, it can execute Online-Sieve for every step~$s'>s$, using the buffered messages, until it catches up with the execution.
It can then proceed normally as per Algorithm~\ref{alg:sieve}.

\textbf{\dpow verification.}
We have assumed thus far that \dpow verification is instantaneous.
In practice, however, it might take some time.
This puts correct nodes at a disadvantage: they have to verify all received messages, whereas Byzantine nodes do not have to verify any messages.
Thus, Byzantine nodes are effectively faster in producing \dpow evaluations.
Our results hold as long as the correct supremacy assumption holds, which we state in terms of the number of \dpow evaluations within a stretch of steps.
The speedup affects only the ratio of Byzantine participation Sieve tolerates once we map number of \dpow evaluations to concrete resources like energy or hardware.

One way to remedy this is for correct nodes to keep upgrading their hardware to increase their power.
Another is to reduce the verification time of a \dpow, ideally to a constant.
This poses an interesting question for applied cryptography: \emph{is there an implementation of our black box \dpow abstraction that can be verified in constant time?}

\textbf{Network synchrony.}
Fully permissionless protocols require synchrony to achieve their guarantees~\cite{lewis-pyePermissionlessConsensus2024}.
In practice, however, all such protocols are deployed in the asynchronous Internet.
The usual way to approximate synchrony in the Internet is to use a gossip protocol and to make that protocol as robust as possible through a variety of mechanisms, {\em e.g.}, by using routers from different networking domains or by blocking nodes that give you bogus messages.
An implementation of Sieve would also rely on such mechanisms, the details of which are outside the scope of this paper.

\textbf{Coffers as pointers.}
Sieve requires coffers, which get prohibitively expensive if they hold actual messages.
In practice, there should be a separate data dissemination layer that makes sure nodes have access to all messages sent thus far, and coffers should contain pointers to messages ({\em e.g.}, message hashes).
The aforementioned gossip protocol can take care of this as well.

\section{Correctness}\label{sec: correctness}

We have specified TTRB as two propositions about the coffers of the messages created by correct nodes.
Correct nodes pick the coffers as the output of one application of either Bootstrap-Sieve or Online-Sieve.
We show that both Bootstrap-Sieve and Online-Sieve maintain an inductive invariant that implies TTRB, hence implying that TTRB is an invariant of Sieve.
We first express the Sieve invariant.

Recall, a timestamp-$s$ message declares~$s$ as its step in its payload, and a message is generated at the step in which the node generating it called~$\textsc{dpow}$ to get a \dpow evaluation.

\begin{definition}[Sieve Invariant]
    For every step~$s\geq 1$ and every correct node~$n$, when~$n$ calls~$\textsc{TTRBDeliver}(s, \mathcal{L})$, the following property~$SI(n, s)$ holds:
    \begin{enumerate}[label=SI\arabic*, ref=SI\arabic*]
        \item\label{SI1} Every~$m\in\mathcal{L}$ is generated at~$s-1$, and
        \item\label{SI2} all correct timestamp-$(s-1)$ messages are in~$\mathcal{L}$.
    \end{enumerate}
\end{definition}

\begin{lemma}\label{lem:SI-implies-TTRB}
    The Sieve Invariant and the Sieve algorithm together imply TTRB.
\end{lemma}
\begin{proof}
    Sieve verifies the \dpow evaluations it receives at each step, which, together with Property~\ref{SI1} of the Sieve Invariant, implies Property~\ref{TTRB1} of TTRB.
    Each correct message $m$ \textsc{TTRBcast} at step $s-1$ declares $s-1$ as its timestamp, and is therefore a timestamp-$(s-1)$ message.
    Together with Property~\ref{SI2} of the Sieve Invariant, this implies Property~\ref{TTRB2} of TTRB.
\end{proof}

We now proceed to prove that what we call the Sieve Invariant (henceforth, SI) holds.
We do so by proving that both Bootstrap-Sieve and Online-Sieve preserve SI inductively: 
($i$) the base case~$SI(n, 1)$ for any correct node~$n$ holds because at step 1 neither Online-Sieve nor Bootstrap-Sieve discard timestamp-0 messages, and 
($ii$) for any~$s\geq 1$ and any correct node~$n$, we assume that~$SI(n', s')$ holds for all~$1\leq s'\leq s-1$ and any correct node~$n'$, and we prove that~$SI(n, s)$ holds.

We start with the simpler case: Online-Sieve.
For a set of messages~$X$, let~$X_s$ be the set of all timestamp-$s$ messages in~$X$.
Let~$C_s$ be the set of all correct timestamp-$s$ messages for all~$s\geq 0$.

\begin{lemma}\label{lem:invariant-implies-supremacy}
    Fix an arbitrary step~$s'$ and an arbitrary correct node~$n$, and let~$X$ be the set of messages received at the start of~$s'$ by~$n$.
    If \,$C_{s'-1}\subseteq X$ and all messages in $X$ are generated at~$s'-1$, then $\texttt{weight}(C_{s'-1}) > (1-\rho)\cdot\texttt{weight}(X)$. 
\end{lemma}
\begin{proof}
    If all messages in~$X$ are generated at~$s'-1$, then they all called the \dpow oracle at step~$s'-1$.
    Moreover, because all messages in~$X$ have been received, their \dpow evaluation also ended in step~$s'-1$.
    Therefore, all messages in~$X$ belong to step~$s'-1$.
    Thus, based on the correct supremacy assumption, correct timestamp-$(s'-1)$ messages are strictly more than a weighted fraction~$1-\rho$ of~$X$.
    We should thus have~$\texttt{weight}(C_{s'-1}) > (1-\rho)\cdot\texttt{weight}(X)$.
\end{proof}

\begin{lemma}\label{lem:online-sieve-preserves-invariant}
    For every correct node $n$ and any step $s\geq 1$, if $SI(n', s')$ holds for all $n'$ and all $s' < s$, then $SI(n, s)$ holds after $n$ executes Online-Sieve at $s$.
\end{lemma}
\begin{proof}
    Consider a correct timestamp-$(s-1)$ message~$m$ and suppose~$n$ is executing the call~$\textsc{OnlineSieve}(s, \mathcal{M}, \mathcal{L})$.
    Note that~$m\in\mathcal{M}$, since~$m$ was produced by a correct node~$n'$ and correct nodes broadcast their messages.
    Since correct nodes set the coffer of their message to be equal to their~$\mathcal{L}$ variable after they have executed Sieve, the set~$\operatorname{coffer}(m)$ satisfies properties~\ref{SI1} and~\ref{SI2} of SI, based on~$SI(n', s-1)$.
    The set~$\mathcal{L}$ also satisfies those properties based on~$SI(n, s-1)$.
    We thus have~$C_{s-2}\subseteq\operatorname{coffer}(m)\cap\mathcal{L}$, which also implies~$\texttt{weight}(C_{s-2})\leq\texttt{weight}(\operatorname{coffer}(m)\cap\mathcal{L})$.
    Moreover, based on Lemma~\ref{lem:invariant-implies-supremacy}, we have~$\texttt{weight}(C_{s-2}) > (1-\rho)\cdot\texttt{weight}(\mathcal{L})$.
    We thus have~$\texttt{weight}(\operatorname{coffer}(m)\cap\mathcal{L}) > (1-\rho)\cdot\texttt{weight}(\mathcal{L})$, which means~$n$ will not discard~$m$.
    We conclude that property~\ref{SI2} of~$SI(n, s)$ holds.

    Now consider any timestamp-$(s-1)$ message~$m'\in\mathcal{M}$ created by some Byzantine node~$n''$, where~$m'$ belongs to a step earlier than~$s-1$, {\em i.e.}, an antique message.
    Since messages in~$C_{s-2}$ were sent only after the end of step~$s-2$, we have~$C_{s-2}\cap\operatorname{coffer}(m') = \emptyset$, because~$n''$ could not have received them.
    Once again, based on Lemma \ref{lem:invariant-implies-supremacy}, we have~$\texttt{weight}(C_{s-2}) > (1-\rho)\cdot\texttt{weight}(\mathcal{L})$, which implies~$\texttt{weight}(C_{s-2}) > \frac{1}{2}\cdot\texttt{weight}(\mathcal{L})$ since~$\rho\leq 1/2$.
    Together, these imply that~$\texttt{weight}(\operatorname{coffer}(m')\cap\mathcal{L})<\frac{1}{2}\cdot\texttt{weight}(\mathcal{L})\leq(1 - \rho)\cdot\texttt{weight}(\mathcal{L})$.
    Therefore,~$m'$ will be discarded, which implies that any message remaining in~$\mathcal{L}$ by the end of the call is generated at~$s-1$, proving property~\ref{SI1} of~$SI(n, s)$.
\end{proof}

Before proceeding to Bootstrap-Sieve, we prove another useful lemma.
For every set of messages~$X$ and every step~$s$, let~$X_{s}^+$ be the subset of~$X$ consisting of all the messages in~$X$ with a timestamp at least~$s$.

\begin{lemma}\label{lem:disjoint-remains-disjoint}
Consider two sets of timestamp-$s$ messages~$S_{s}^1$ and~$S_{s}^2$ and assume that~$X^1$ and~$X^2$ are two DAGs consistent with~$S_{s}^1$ and~$S_{s}^2$, respectively.
Then~$S_{s}^1 \cap S_{s}^2 =\emptyset$ implies $X^1\cap X^2 =\emptyset$.
\end{lemma}
\begin{proof}
    Since~$X^1$ and~$X^2$ are DAGs consistent with~$S_{s}^1$ and~$S_{s}^2$, respectively, we thus have~$X_{s}^1 = S_{s}^1$ and $X_{s}^2 = S_{s}^2$.
    We show that, for every natural number~$s'\geq s$, if~$X_{s'}^1 \cap X_{s'}^2 =\emptyset$ then~$X_{s'+1}^1 \cap X_{s'+1}^2 =\emptyset$.
    The lemma then follows by induction.

    Consider~$s'\geq s$ such that~$X_{s'}^1 \cap X_{s'}^2 =\emptyset$ and assume toward a contradiction that there is some~$m$ such that~$m\in X_{s'+1}^1 \cap X_{s'+1}^2$.
    Because~$X^{1}$ is a consistent DAG,~$X_{s'}^{1}$ is strictly more than a weighted fraction~$1-\rho$ of~$m$'s coffer.
    Similarly,~$X_{s'}^{2}$ is strictly more than a weighted fraction~$1-\rho$ of~$m$'s coffer.
    Since~$\rho\leq 1/2$, there must be a message~$m'$ such that~$m'\in X_{s'}^{1}\cap X_{s'}^{2}$.
    This contradicts our assumption that~$X_{s'}^1 \cap X_{s'}^2 =\emptyset$.
    We thus have~$X_{s'+1}^1 \cap X_{s'+1}^2 = \emptyset$.
\end{proof}

We are now ready to tackle Bootstrap-Sieve.

\begin{lemma}\label{lem:Bootstrap-sieve-preserves-invariant}
    For every correct node $n$ and any step $s\geq 1$, if $SI(n', s')$ holds for all $n'$ and all $s' < s$, then $SI(n, s)$ holds after $n$ executes Bootstrap-Sieve at $s$.
\end{lemma}
\begin{proof}
    Suppose that~$n$ is executing the call~$\textsc{BootstrapSieve}(s, \mathcal{M})$, which proceeds in iterations.
    There is an elegant connection between these iterations and SI: the iterations preserve SI within the history maintained by~$n$ during the call execution.
    In other words, for every~$1\leq s'\leq s-1$, the following hold at the end of iteration~$s'$:
    \begin{enumerate}[label=IH\arabic*$(s')$, ref=IH\arabic*$(s')$, left=0pt]
        \item\label{IH1} All correct timestamp-$s'$ messages are in~$\mathcal{L}$, and
        \item\label{IH2} every~$m\in\mathcal{L}_{s'}$\,is generated at a step greater than or equal to~$s'$.
    \end{enumerate}
    Note that IH$(s-1)$ implies $SI(n, s)$.
    It thus suffices to prove properties~\ref{IH1} and~\ref{IH2} using induction on the iteration~$s'$.

    \noindent\textbf{Base case.} Correct nodes broadcast their timestamp-$0$ messages, so they are all in~$\mathcal{M}$.
    Bootstrap-Sieve does not discard any timestamp-$0$ messages, therefore, IH1$(0)$ holds.
    IH2$(0)$ holds as messages cannot have a negative generation time.
    
    \noindent\textbf{Induction Hypothesis.} IH1$(s'')$ and IH2$(s'')$ hold for all~$1\leq s''\leq s'-1$.

    \noindent\textbf{Induction step.} We have to prove that \ref{IH1} and \ref{IH2} hold.
    We first prove that~$C_{s'-1}^+$ is a DAG consistent with~$C_{s'-1}$ within~$\mathcal{L}$.

    Based on IH1$(s'-1)$, we have~$C_{s'-1}\subseteq\mathcal{L}$.
    Since Bootstrap-Sieve discards timestamp-$r$ messages from~$\mathcal{L}$ only at iteration~$r$, we also have~$C_r\subseteq\mathcal{L}$ for~$r\geq s'$.
    Now, for all~$r\geq s'-1$ and any~$m\in C_{r+1}$ created by some correct node~$n'$, according to~$SI(n', r)$ and the Sieve algorithm we have~$C_r\subseteq\operatorname{coffer}(m)$, and also every message in~$\operatorname{coffer}(m)$ is generated at~$r$.
    Applying Lemma~\ref{lem:invariant-implies-supremacy} to~$C_r$ and~$\operatorname{coffer}(m)$ implies that~$C_r$ is strictly more than a weighted fraction~$1-\rho$ of~$\operatorname{coffer}(m)$, which establishes that~$C_{s'-1}^+ = \cup_{r\geq s'-1}C_r$ is a DAG consistent with~$C_{s'-1}$ within~$\mathcal{L}$.
    
   We now prove \ref{IH1}.
    Consider any correct message~$m\in C_{s'}$, and let, for any~$r$,~$TS(r)$ be the set of all messages generated at a step greater than or equal to~$r$ in~$\mathcal{L}$.
    Let~$C'$ be one of the heaviest timestamp-$(s'-1)$ consistent DAGs containing~$m$.
    Since~$C_{s'-1}^+$ is a timestamp-$(s'-1)$ consistent DAG containing~$m$, then~$C'$ exists and we have~$\texttt{weight}(C')\geq\texttt{weight}(C_{s'-1}^+)$.
    Based on the correct supremacy assumption, we have~$\texttt{weight}(C_{s'-1}^+) > (1-\rho)\cdot\texttt{weight}(TS(s'-1)) > \frac{1}{2}\cdot\texttt{weight}(TS(s'-1))$ (note that~$\rho\leq 1/2$), which allows us to further deduce~$\texttt{weight}(C') > \frac{1}{2}\cdot\texttt{weight}(TS(s'-1))$.
    Therefore, for any~timestamp-$(s'-1)$ consistent DAG~$C''$ such that~$C''\cap C' = \emptyset$, since both~$C'\subseteq TS(s'-1)$ and~$C''\subseteq TS(s'-1)$ hold based on the induction hypothesis IH, we must have~$\texttt{weight}(C'') < \texttt{weight}(C')$.
    We conclude that Bootstrap-Sieve does not discard~$m$, which completes our proof of \ref{IH1}.

    Let us finally prove \ref{IH2}.
    Consider an antique message~$m'$ generated at a step less than~$s'$.
    Since correct timestamp-$(s'-1)$ messages were only ready at the end of step~$s'-1$, $m'$ does not contain any of them in its coffer, {\em i.e.},~$C_{s'-1}\cap\operatorname{coffer}(m')=\emptyset$.
    Let~$C'$ be any heaviest timestamp-$(s'-1)$ consistent DAG containing~$m'$.
    Based on Lemma~\ref{lem:disjoint-remains-disjoint}, we get~$C_{s'-1}^+\cap C' = \emptyset$.
    Based on IH2$(s'-1)$, every timestamp-$(s'-1)$ message in the seed of~$C'$ is generated at a step greater than or equal to~$s'-1$; therefore, all messages in~$C'$ are generated at a step greater than or equal to~$s'-1$ because they were generated after the messages in the seed, {\em i.e.},~$C'\subseteq TS(s'-1)$.
    We have already established that~$C_{s'-1}^+\subseteq TS(s'-1)$ and that~$\texttt{weight}(C_{s'-1}^+) > \frac{1}{2}\cdot\texttt{weight}(TS(s'-1))$.
    We thus have~$\texttt{weight}(C') < \texttt{weight}(C_{s'-1}^+)$, which means~$n$ will discard~$m'$.
    This concludes the proof for \ref{IH2}.    
\end{proof}

\begin{theorem}
    TTRB is an invariant of the Sieve algorithm.
\end{theorem}
\begin{proof}
    For every correct node~$n$,~$SI(n, 1)$ holds because at step 1 neither Online-Sieve nor Bootstrap-Sieve discard timestamp-0 messages, and according to the correct supremacy assumption correct timestamp-0 messages are strictly more than a weighted fraction~$1-\rho$ of messages generated at step 0.
    Based on this, an inductive application of Lemma \ref{lem:online-sieve-preserves-invariant} and Lemma \ref{lem:Bootstrap-sieve-preserves-invariant}, in addition to the fact that Sieve picks the coffer of the message it wants to send running either Online-Sieve or Bootstrap-Sieve, implies that SI is an invariant of Sieve.
    Lemma~\ref{lem:SI-implies-TTRB} concludes our proof.
\end{proof}

\section{Sieve-MMR: Fully-Permissionless Total-Order Broadcast}\label{sec: FP-TOB}

Fully-permissionless TOB is now within reach: Sieve implements TTRB, and TTRB is sufficient to enable the MMR protocol by Malkhi et al.~\cite[Appendix A]{malkhiByzantineConsensusFully2023} to operate correctly. To port MMR to a fully permissionless model, we build Sieve-MMR by layering MMR atop Sieve, which provides the essential message delivery guarantees that MMR requires.
Specifically, MMR relies on the following:

\begin{assumption}
    \label{assumption:ttrb-for-tob}
    For every correct node~$n$, for every step~$s\geq 1$, if TTRB calls~$\textsc{TTRBDeliver}(s, \mathcal{L})$ at node~$n$ at the start of step~$s$, then:
    \begin{enumerate}[label=MMR\arabic*, ref=MMR\arabic*, left=0pt]
        \item\label{MMR1} Messages TTRBcast by correct nodes in step~$s-1$ account for strictly more than a weighted fraction~$2/3$ of the total weight of the messages in~$\mathcal{L}$.
        \item\label{MMR2}Every message TTRBcast by a correct node in step~$s-1$ appears in~$\mathcal{L}$.
    \end{enumerate}
\end{assumption}
TTRB with parameter~$\rho=1/3$ provides this guarantee: Property~\ref{TTRB2} immediately implies Property~\ref{MMR2} above, and Property~\ref{TTRB1} implies the Property~\ref{MMR1} according to Lemma~\ref{lem:invariant-implies-supremacy} from Section~\ref{sec: correctness}.
As a result, Sieve-MMR inherits MMR's ability to tolerate a~$1/3$-bounded adversary\footnote{Standalone Sieve tolerates a~$1/2$-bounded adversary.}.

Building on Assumption~\ref{assumption:ttrb-for-tob}, Sieve-MMR implements TOB as specified in Section~\ref{sec: model}.
In a nutshell, nodes accept blocks \emph{submitted} by external clients and order them in a growing chain that they periodically \emph{commit} to their clients, at which point all its prefixes are considered \emph{committed}.

Sieve-MMR retains the TOB guarantees (\S\ref{sec:tob-def}):
    \begin{itemize}
        \item {\em Consistency: If two correct nodes commit chains $\Lambda_1$ and $\Lambda_2$, then $\Lambda_1$ and $\Lambda_2$ are compatible.}
        \item {\em Progress: Let $\Lambda$ be the longest chain committed by all correct nodes.
        At all times, with probability 1, $\Lambda$ eventually includes at least one more block submitted by a correct node.}
    \end{itemize}

Additionally, we care about the time it takes for a transaction to become final, \emph{i.e.}, to appear in a committed chain.
For brevity, let us refer to a block submitted by a correct node as a {\em correct block}. Then, a proxy for finality is the number of steps necessary, starting from some step~$s$, for all active correct nodes to commit a new correct block (\emph{i.e.}, the commit latency).
Sieve-MMR inherits from MMR the following commit latency guarantees:

    \begin{enumerate}[label=CL\arabic*, ref=CL\arabic*, left=0pt] {\em
        \item In the best case, all correct active nodes commit a new correct block 3 steps after it was submitted.
        \item In general, in expectation, all correct active nodes commit a new correct block~$7$ steps after it was submitted.}
    \end{enumerate}

The remainder of this section is dedicated to describing Sieve-MMR; we provide a correctness proof in the Appendix.

\subsection{The Sieve-MMR Algorithm}
\label{sec:fp-mmr}

In Sieve-MMR, correct nodes implement TTRB using Sieve and run the MMR algorithm on top of Sieve.
The MMR algorithm prescribes, at each step~$s$, which messages to TTRBCast through Sieve in response to the set of messages~$\mathcal{L}_{s-1}$ delivered by Sieve (see~\Cref{fig:sieve-mmr-stack}).
The composition of Sieve and the MMR algorithm is what we call Sieve-MMR.

For the rest of this section, when we say that a node~$n$ receives a message~$m$ in a step~$s$, we mean that~$m$ belongs to the set~$\mathcal{L}$ delivered by Sieve at~$n$ in step~$s$.
Moreover, when we say that a node sends or broadcasts a message~$m$, we mean that it calls TTRBCast$(m)$.

Sieve-MMR inherits its consensus logic almost verbatim from the MMR algorithm.
For safety, it relies solely on Assumption~\ref{assumption:ttrb-for-tob}, guaranteed by TTRB.
For liveness, Sieve-MMR relies on a probabilistic leader-election component, accessible locally at each node, that, at each step~$s$, determines a leader message among all messages received in step~$s$.
This leader-election component must guarantee that, at each step $s$, with probability strictly greater than 2/3, all active correct nodes obtain the same leader $l$ and $l$ was sent by a correct node in step~$s-1$.
While MMR implements this component using verifiable random functions~\cite{micaliVerifiableRandomFunctions1999}, Sieve-MMR relies on the assumption that the \dpow is a random oracle (\S\ref{sec:leader}).

Next, we describe the MMR algorithm (\Cref{algo:mmr}).
The algorithm is called by TTRB in each new step through a $\textsc{TTRBDeliver}$ upcall.
The MMR algorithm classifies each step as either a proposal step, if the step number is even, or a commit step, if the step number is odd.
In a proposal step, each node broadcasts a message that conveys a vote for a chain and a proposal for a (possibly longer) chain.
In a commit step, each node may commit a new chain and broadcasts a message that conveys a vote for a (possibly longer) chain.
In both types of steps, a vote for a chain $\Lambda$ also counts as a vote for all prefixes of $\Lambda$.

Before we describe the rules that nodes follow to vote for, propose, and commit chains, we need the notions of maximal chains, (maximal) grade-0 chains, and (maximal) grade-1 chains.
\begin{definition}[Grade-0 and grade-1 chains]
    We say that a chain $\Lambda$ has grade~1 at a node $n$ when, among the votes received by $n$ in the current step, the votes for extensions of $\Lambda$ account for strictly more than 2/3 of the proof-of-work weight.
    We say that a chain $\Lambda$ has grade~0 at a node $n$ when, among the votes received by $n$ in the current step, the votes for extensions of $\Lambda$ account for strictly more than 1/3 of the proof-of-work weight.
\end{definition}
Note that it follows that a chain that has grade~1 also has grade~0, but the converse is not true; moreover, the empty chain always has both grades~0 and~1.
\begin{definition}[Maximal chains]
    We say that a chain $\Lambda$ is maximal among a set of chains if no chain in the set is a strict extension of $\Lambda$ (note that two different, incompatible chains can both be maximal in the same set).
    We say that a chain $\Lambda$ is a maximal grade-1 chain (or maximal grade-0 chain) at $n$ when $\Lambda$ is maximal among the set of chains that have grade~1 (respectively, grade~0) at $n$.
\end{definition}

We are now ready to describe the algorithm in full.
In each step~$s$, each node $n$~must proceed as follows:
\begin{itemize}
    \item If~$s=0$ (this is the first step),~$n$ votes for the empty chain and proposes a chain consisting of an arbitrary submitted block (we assume each correct node always has at least one fresh submitted block available).
    \item If~$s=2k+1$ for some~$k\geq 0$, then~$s$ is a commit step and~$n$ consults the leader-election oracle, obtains a leader~$l$, and, if~$l$ carries a proposal~$\Lambda_l$ and~$\Lambda_l$ extends~$n$'s maximal grade-0 chain,~$n$ votes for~$\Lambda_l$; otherwise,~$n$ votes for the maximal grade-0 chain\footnote{We show in~\Cref{lem:unique-local-one-third} that maximal grade-0 chains are locally unique in commit steps.}.
        Moreover,~$n$ commits the maximal grade-1 chain.
    \item If~$s=2k$ for some~$k>0$, then~$s$ is a proposal step and~$n$ votes for the maximal grade-1 chain\footnote{\Cref{lem:two-thirds-compatible} implies that maximal grade-1 chains are locally unique.}.
        Moreover,~$n$ selects a submitted block not from its last committed chain, appends it to a randomly chosen maximal grade-0 chain\footnote{\Cref{lem:liveness-key} shows that there may be at most two maximal grade-0 chains.}, and proposes the resulting chain.
\end{itemize}
The Appendix includes a detailed correctness proof.

Note that the MMR algorithm is stateless: except for the set of blocks submitted by clients, the algorithm's actions only depend on the set of messages received in the current step.
Thus, each node can become active or inactive at any step without compromising MMR's properties.

\begin{algorithm}
    \caption{The MMR algorithm, code for node $n$.}%
    \label{algo:mmr}
    \begin{algorithmic}[1]
        \LeftComment{State variables:}
        \State $\mathcal{B}\gets\emptyset$ \Comment{\parbox[t]{0.7\linewidth}{Set of non-committed blocks submitted by clients (updated by upper layer)}}
        \Procedure{TTRBDeliver}{$s$, $\mathcal{L}$}\Comment{Upcall from TTRB}
            \If{$s = 0$}\Comment{A proposal step}
                \State{$b\gets \text{an element of}\ \mathcal{B}$}
                \State{$m\gets \left[  \text{proposal}: b,  \text{vote}: \left\langle  \right\rangle \right]$}
                \State \Return{$m$}\Comment{Return $m$ to TTRB}
            \EndIf
            \If{$s = 2k+1$ for $k\geq 0$}\Comment{A commit step}
                \State{$m_{l}\gets \Call{ElectLeader}{\mathcal{L}}$}
                \State{$\Lambda_{0}\gets \text{the maximal grade-0 chain in }\mathcal{L}$}
                \If{$m_l$ carries a proposal $\Lambda_l$ and $\Lambda_l$ extends $\Lambda_{0}$}
                    \State{$m\gets \left[ \text{vote}:\Lambda_{l} \right]$}
                \Else
                    \State{$m\gets \left[ \text{vote}:\Lambda_{0} \right]$}
                \EndIf
                \State{$\Lambda_1\gets \text{the maximal grade-1 chain in }\mathcal{L}$}
                \State{\Call{Commit}{$\Lambda_{1}$}}\Comment{\parbox[t]{0.45\linewidth}{Upcall to the consensus module\vspace{1mm}}}
                \State{$\mathcal{B}\gets \mathcal{B}\setminus\operatorname{blocks}(\Lambda_{1})$}\Comment{\parbox[t]{0.4\linewidth}{Remove committed blocks from $\mathcal{B}$\vspace{1mm}}}
                \State{\Return{$m$}}\Comment{Return $m$ to TTRB}
            \EndIf
            \If{$s = 2k$ for $k> 0$}\Comment{A proposal step}
                \State{$\Lambda_{0}\gets \text{a random maximal grade-0 chain in }\mathcal{L}$}
                \State{$\Lambda_1\gets \text{the maximal grade-1 chain in }\mathcal{L}$}
                \State{$b\gets \text{an element of } \mathcal{B}$}
                \State{$\Lambda\gets \Lambda_0,b$}\Comment{Append $b$ to $\Lambda_{0}$}
                \State{$m\gets \left[ \text{vote}:\Lambda_{1},\text{proposal}:\Lambda \right]$}
                \State{\Return{$m$}}\Comment{Return $m$ to TTRB}
            \EndIf
        \EndProcedure
    \end{algorithmic}
\end{algorithm}

\subsection{Leader Election}
\label{sec:leader}

Each message delivered by Sieve to MMR is of the form $\langle m, dpow,w \rangle$, where $dpow$ is a \dpow evaluation, which we assume to be a random oracle.
Each round, each correct node $n$ picks a leader message as follows.
For each message $\langle m,dpow,w \rangle$ received by $n$ in the current step, $n$ creates $w$ tokens $H(dpow)$, $H(dpow+1)$, ..., $H(dpow+w-1)$ where $H$ is a random oracle ({\em e.g.}, a cryptographic hash function).
Then, $n$ selects as leader the message $m$ associated with the largest token among all tokens generated for all messages (assuming no collision).

By the correct supremacy assumption, with probability 2/3, a correct node has the largest token, and since all correct nodes receive all messages from all correct nodes of the previous step, with probability 2/3, all correct nodes agree on their leader message.
Depending on message weights, we may have to create a large number of tokens: Swiper~\cite{tonkikhSwiperNewParadigm2024} proposes algorithms to reduce the number of tokens needed.

\section{Implementing Deterministic Proof-of-Work}
\label{sec:det-pow}

In this section, we briefly present a concrete implementation of our \dpow primitive; our results rely on the black-box \dpow guarantees and any such implementation would work.
The construction is due to Coelho~\cite{coelhoAlmostConstantEffortSolutionVerification2008}.
We assume nodes have access to a random oracle function~$\mathcal{H}$ mapping binary strings of arbitrary length to binary strings of length~$\lambda$, for some security parameter~$\lambda$.
In practice, one can use SHA-256.

The implementation consists of two algorithms: an algorithm~$\mathcal{P}$ to generate a proof-of-work and an algorithm~$\mathcal{V}$ for verifying a proof-of-work.
Both algorithms are parameterized by a security parameter~$k$.
The algorithm~$\mathcal{P}$ takes as input a challenge~$\chi$ and an integer weight~$w$, and returns a proof~$dpow$; the verification algorithm~$\mathcal{V}$ takes as input a proof~$dpow$, a challenge~$\chi$, and a weight~$w$, and returns a boolean indicating whether the proof is valid or not.

The two algorithms guarantee that:
\begin{itemize}
    \item For every~$\chi$ and~$w\geq k$, a node that faithfully executes~$\mathcal{P}$ on inputs~$\chi$ and~$w$ makes~$2w+k$ calls to the random oracle and outputs a proof~$dpow$ such that~$\mathcal{V}(dpow, \chi, w)$ returns true.
    \item For every~$dpow$,~$\chi$, and~$w\geq k$, a node that faithfully executes~$\mathcal{V}$ on inputs~$dpow$,~$\chi$, and~$w$ makes~$k\log w$ calls to the random oracle.
    \item For every~$\chi$, $w\geq k$, and~$0<t\leq1$, if a node creates a proof~$dpow$ by calling the random oracle less than~$t(2w+k)$ times, then~$\mathcal{V}(dpow, \chi, w)$ returns true with probability less than~$t^k$.
\end{itemize}
Algorithm~$\mathcal{P}$ works as follows.
Given a challenge~$\chi$ and a weight~$w$, the algorithm computes a Merkle tree commitment~$\Phi$ to the~$w$ leaves~$l_{1}=\mathcal{H}(\chi)$, $l_{2}=\mathcal{H}(\chi+1)$, \dots, $l_{w}=\mathcal{H}(\chi+w-1)$.
Let~$\phi$ be the root of this tree.
The algorithm determines~$k$ distinct leaf indices by computing the natural number~$\lceil k\mathcal{H}(\phi+i)/2^{\lambda}\rceil$, starting with~$i=0$ and incrementing~$i$ until it obtains~$k$ distinct natural numbers.
The proof~$dpow$ then consists of the~$k$ Merkle paths corresponding to the~$k$ indices computed above,  plus the root~$\phi$.

Algorithm~$\mathcal{V}$ works as follows.
The algorithm first computes the~$k$ indices in the same way as in algorithm~$\mathcal{P}$, and then it checks that the provided Merkle paths are indeed correct Merkle paths corresponding to the~$k$ indices.

In practice, we must pick a concrete number of Merkle paths~$k$ that must be revealed by provers.
A larger~$k$ increases proof size and decreases the probability that an adversary that does not compute the full tree will create a proof that is accepted by some correct nodes.
Given a target security parameter~$p$ and a minimum-work threshold~$0<t\leq1$, we can ensure that no adversary produces a valid proof of work with probability higher than~$2^{-p}$ using fewer than~$n=t(2w+k)$ queries to the random oracle by choosing~$k$ such that~$t^{k}<2^{-p}$.
Practical deployments should then consider that Byzantine nodes are faster in producing \dpow evaluations by a factor of $1/t$.
This means that an adversary that is $1/3$-bounded in the model of~\Cref{sec: model} will in reality spend only a fraction $t/3$ of the energy or other real-world resources spent by all nodes in the system.

\section{Related Work}\label{sec: related}

{\bf Permissionless settings.}
Lewis-Pye and Roughgarden~\cite{lewis-pyePermissionlessConsensus2024} formally classify permissionless systems in three settings: ($i$) the {\em quasi-permissionless setting} ({\em e.g.}, Tendermint~\cite{buchmanLatestGossipBFT2019} or Algorand~\cite{Algorand}); ($ii$) the {\em dynamically-available setting} ({\em e.g.}, Ouroboros~\cite{kiayiasOuroborosProvablySecure2017}); and ($iii$) the {\em fully-permissionless setting} {\em (e.g.}, Bitcoin~\cite{nakamotoBitcoinPeertopeerElectronic2008}).
The quasi-permissionless setting and the dynamically-available setting model proof-of-stake systems, which track their participants on chain and typically assume that more than 1/2 or 2/3 are correct.
The first assumes always active correct participants, while the second allows them to be inactive as long as the remaining active correct participants still form a supermajority.
The sleepy model~\cite{pass_sleepy_2017} is similar to the dynamically-available setting, but with a static list of nodes.
The fully-permissionless setting further generalizes the dynamically-available setting by assuming no knowledge of participation.
Like Bitcoin, Sieve-MMR is fully-permissionless.

{\bf Sleepy/dynamically-available protocols.}
Sleepy consensus~\cite{pass_sleepy_2017} was the first consensus protocol to allow inactive correct participants.
Several other consensus protocols followed~\cite{daianSnowWhiteRobustly2019,kiayiasOuroborosProvablySecure2017,damatoGoldfishNoMore2023}, all guaranteeing safety and liveness probabilistically, until Momose and Ren~\cite{momose_constant_2022} achieved deterministic safety.
Later deterministically-safe protocols~\cite{malkhi_towards_2023,gafni_brief_2023} achieve a latency of a few message delays.
Sieve-MMR borrows the consensus logic of MMR, a deterministically-safe, dynamically-available TOB protocol~\cite[Appendix A]{malkhiByzantineConsensusFully2023}, and ports it to the fully-permissionless setting.

{\bf Mitigations against long-range attacks.}
PoS systems are vulnerable to long-range attacks that cause safety violations at little cost to the attacker.
Using VDFs (\emph{e.g.},~\cite{debPoSATProofofWorkAvailability2021a, Posh, long2019nakamoto, fairledger,chia_blockchain}) or ephemeral keys (\emph{e.g.},~\cite{Algorand, azouviWinkleFoilingLongRange2020}) is effective against posterior-corruption long-range attacks.
Nevertheless, in the absence of external trust assumptions, long-range attacks always prevent attaining slashable safety~\cite{tasBitcoinEnhancedProofofStakeSecurity2023}.
Babylon and Pikachu~\cite{azouviPikachuSecuringPoS2022,tasBitcoinEnhancedProofofStakeSecurity2023} prevent long-range attacks by checkpointing their state onto Bitcoin, which is an external trusted component.
Sieve-MMR is a PoW protocol and is resilient against long-range attacks.
Budish et al.~\cite{budishEconomicLimitsPermissionless2024} study the security of PoS, including long-range attacks, from an economic perspective.

{\bf Proof-of-work protocols.}
PoW protocols operate in the fully-permissionless setting and do not suffer from long-range attacks.
A line of work has improved PoW throughput~\cite{eyal2016ng,fitziParallelChainsImproving2018} and latency~\cite{bagaria2019prism,garay2024constant}, with Garay et al.~\cite{garay2024constant} achieving expected constant latency.
However, the protocol of Garay et al. relies on high-variance probabilistic building blocks that cannot be sensibly analyzed in a deterministic model.

Gorilla~\cite{Gorilla}, a BFT sequel to the Sandglass protocol~\cite{Sandglass}, achieves deterministic safety using verifiable delay functions as a PoW primitive, but it has a latency exponential in the number of participants.
Gorilla and Sieve-MMR rely on entirely different mechanisms to achieve consensus.
Gorilla executions proceed in asynchronous rounds, where a node proceeds to the next round if it receives a certain threshold of messages.
Then, inspired by  Ben-Or's protocol~\cite{ben-or}, the node proposes a value~$v$ in the next round if all of the messages it received in the previous round proposed~$v$ unanimously; otherwise, it picks a random value.
The node decides~$v$ if it keeps receiving messages unanimously proposing~$v$ for a sufficiently long sequence of  rounds---a fortunate event that is guaranteed to happen with positive probability.
However, since both the threshold of messages required to move to the next round, and the length of the sequence of unanimous rounds required to decide are proportional to the upper bound on the number of nodes, the probability of this fortunate event (and thus Gorilla average latency)  is exponential in that same upper bound. 
Sieve-MMR instead relies on quorum intersection arguments and, for them to work,  depends on the correct supremacy assumption.
These arguments do not depend on the number of messages, nor do they rely on the execution having produced a sufficient number of messages.

Keller and Böhme~\cite{kellerParallelProofofWorkConcrete2023a} propose a TOB protocol \(\mathcal{B}_k\) consisting of a sequence of instances of a consensus protocol \(\mathcal{A}_k\).
In \(\mathcal{A}_k\), roughly speaking, each node casts votes by solving Bitcoin-style probabilistic PoW puzzles and votes for the value that currently has the most votes; nodes decides on a value when it reaches \(k\) votes. Compared to Sieve-MMR, a disadvantage of this  approach is that, like Bitcoin, it cannot be meaningfully analyzed in a deterministic model.
Moreover, again like Bitcoin, the latency of this protocol depends on the desired level of security.
On the flip side, thanks to the probabilistic PoW puzzle, in expectation only a few nodes send concurrent messages for each consensus decision; in contrast, Sieve-MMR uses all-to-all communication at each protocol step.

\section{The Road Ahead}\label{sec: road-ahead}
This paper presents Sieve-MMR, a TOB protocol that achieves deterministic safety and constant expected latency in the fully permissionless model.
Sieve-MMR is composed of two layers: Sieve and MMR.
Our main contribution is Sieve, a novel algorithm that implements a novel broadcast primitive, TTRB.
TTRB enables the MMR protocol to operate in the fully permissionless model by providing the assumptions it typically relies on in the more restrictive dynamically available model. This work opens two promising directions for future research.

\textbf{\emph{Practical}, fast, and secure PoW consensus.} Sieve brings us to the threshold of a practical protocol, but main challenges remain: the exponential complexity of Bootstrap-Sieve when implemented naively, and the verification overhead of \dpow{}s. Bootstrap-Sieve, as presented here, functions more as a specification than a fully realized algorithm---it is largely declarative. Developing an efficient implementation would elevate Sieve from a theoretical construct to a protocol suitable for practical deployment. Regarding \dpow verification, it is highly advantageous for correct nodes to verify received messages in batches and within a short time window. Achieving this level of efficiency may require additional cryptographic tools, such as zero-knowledge proofs~\cite{zkp}.

\textbf{Porting PoS protocols to the PoW setting.}  Given the surgical nature of our construction, we conjecture that TTRB may serve as a general mechanism for porting other dynamically available protocols to the fully permissionless model. This raises an intriguing open question: is TTRB a canonical bridge between the dynamically available and fully permissionless models?

\section*{Acknowledgments}
We are grateful to the Sui Foundation for their generous support of this research.

\bibliography{sample-base}

\appendix

\section{Correctness of the Sieve-MMR Algorithm}
\label{sec:sieve-mmr-correctness}

In this section, we prove the correctness of the Sieve-MMR algorithm.
Since, in~\Cref{sec: correctness}, we have shown that Sieve implements TTRB, in this section we show that, assuming a correct TTRB implementation, the MMR algorithm implements~TOB.

To simplify the terminology, we say that a message is sent when it appears as an argument to a \textsc{TTRBCast} downcall to Sieve, and that it is received when it appears as an argument of a \textsc{TTRBDeliver} upcall from Sieve.

\subsection{Key Properties of TTRB}

We start with three key properties that stem directly from the guarantees of TTRB expressed in Assumption~\ref{assumption:ttrb-for-tob}.
Then we will show that these three properties imply the correctness of MMR when run on top of Sieve.

\begin{property}
    \label{lem:two-thirds-is-maj-correct}
    Consider a step~$s > 0$ and a correct node $n$ active in step~$s$.
    Assume that $M$ is a set of messages consisting of strictly more than two thirds (by weight) of the messages that $n$ receives in step~$s$.
    Then, $M$ includes a strict majority (by weight) of the correct messages sent in step~$s-1$.
\end{property}
\begin{proof}
    Let $\mathcal{L}$ be the set of messages delivered to $n$ in step $s$ and let $\mathcal{C}$ be the set of all correct timestamp-$(s-1)$ messages sent.
    With $\rho=1/3$ and Assumption~\ref{assumption:ttrb-for-tob}, we have that $\mathcal{C} \subseteq \mathcal{L}$ and $\mathcal{C}$ accounts for at least two thirds (by weight) of the messages in $\mathcal{L}$.
    Hence, if $M$ also consists of strictly more two thirds (by weight) of $\mathcal{L}$, then $\mathcal{C}\cap M$ is a strict majority (by weight) of $\mathcal{C}$.
\end{proof}

\begin{property}
    \label{lem:maj-correct-is-one-third}
    Consider a step~$s>0$ and a correct node $n$ active in step~$s$.
    Assume that $M$ is a set of messages consisting of a strict majority (by weight) of the correct messages sent in step~$s-1$.
    Then, for every node $n$ active in step~$s$, $M$ consists of strictly more than one-third (by weight) of the messages that $n$ receives in step~$s$.
\end{property}

\begin{property}
    \label{lem:one-third-includes-correct}
    Consider a step~$s>0$ and a correct node $n$ active in step~$s$.
    Assume that $M$ is a set of messages consisting of strictly more than one third (by weight) of the messages that $n$ receives in step~$s$.
    Then $M$ includes a message sent by a correct node in step $s-1$.
\end{property}

\subsection{Key Protocol Lemmas}

Next, we prove three lemmas that embody the key principles used in the MMR algorithm.
Once these lemmas are established, the rest of the correctness proofs are almost routine.

In each of the three lemmas, we consider a single step~$s$.
The first key lemma states that grade-1 chains are always compatible:
\begin{lemma}
    \label{lem:two-thirds-compatible}
    Consider a step~$s> 0$, two nodes $n$ and $p'$, and two chains $\Lambda$ and $\Lambda'$ such that $\Lambda$ has grade~1 at $n$ and $\Lambda'$ has grade~1 at~$p'$.
    Then $\Lambda$ and $\Lambda'$ are compatible.
\end{lemma}
\begin{proof}
    By~\Cref{lem:two-thirds-is-maj-correct}, we have that a strict majority (by weight) of the correct messages sent in step~$s-1$ votes for an extension of $\Lambda$.
    Similarly, a strict majority (by weight) of the messages sent in step~$s-1$ votes for an extension of $\Lambda'$.
    Since two strict majorities must intersect, we obtain a correct message sent in step~$s-1$ that votes for a chain $\Lambda''$ that is an extension of both $\Lambda$ and~$\Lambda'$.
    Thus $\Lambda$ and $\Lambda'$ are compatible.
\end{proof}
Note that~\Cref{lem:two-thirds-compatible} implies that, for each node~$n$ and step~$s$, there is a unique maximal grade-1 chain at $n$ in step~$s$.
This justifies our use of ``the maximal grade-1 chain'' in~\Cref{algo:mmr}.

Next, we turn to the second key lemma: if all correct nodes vote for compatible chains and if a chain $\Lambda$ has grade~1 at some node, then all chains that are maximal with grade~0 at any node are extensions of $\Lambda$.
\begin{lemma}
    \label{lem:one-third-extends-two-thirds}
    Consider a step~$s> 0$ and assume that all correct nodes active in step~$s-1$ vote for compatible chains.
    Consider two nodes $n$ and $n'$, and two chains $\Lambda$ and $\Lambda'$ such that $\Lambda$ has grade~1 at $n$ and $\Lambda'$ is maximal with grade~0 at $n'$.
    Then $\Lambda$ is a prefix of $\Lambda'$.
\end{lemma}
\begin{proof}
    First, note that, since all correct nodes active in step~$s-1$ vote for compatible chains, by~\Cref{lem:one-third-includes-correct}, all chains with grade 0 at a correct node in step $s$ are compatible and thus there is a unique maximal grade-0 chain at $n'$ in step $s$.
    Moreover, since $\Lambda$ has grade~1 at $n$ in step~$s$, by~\Cref{lem:two-thirds-is-maj-correct}, a strict majority (by weight) of the correct messages sent in step~$s-1$ votes for an extension of $\Lambda$.
    Thus, by~\Cref{lem:maj-correct-is-one-third}, $\Lambda$ has grade (at least) 0 at $n'$ in step~$s$.
    Finally, since $\Lambda'$ is the (unique) maximal grade-0 chain at $n'$ in step~$s$, we have that $\Lambda$ is a prefix of~$\Lambda'$.
\end{proof}

The third and last key lemma states that each correct node $n$ has at most two maximal grade-$0$ chains $\Lambda_1$ and $\Lambda_2$, and that there is a unique chain $\Lambda\in\{\Lambda_1,\Lambda_2\}$ such that, for every active correct node $n'$, if $\Lambda'$ has grade~1 at $n'$ then $\Lambda'$ is a prefix of $\Lambda$ (note that the order of quantification is important here: it is the same $\Lambda$ for every $n'$).
\begin{lemma}
    \label{lem:liveness-key}
    Consider a step~$s>0$ and a correct node $n$ active in step~$s$.
    There are at most two maximal grade~0 chains at $n$ in step~$s$ and, if $\{\Lambda_1,\Lambda_2\}$ is the set of maximal grade~0 chains at $n$ in step~$s$ (possibly $\Lambda_1=\Lambda_{2}$), then there is a chain $\Lambda\in \{\Lambda_1,\Lambda_2\}$ such that, for every correct node $n'$ active in step~$s$, if $\Lambda'$ is maximal with grade~1 at $n'$, then $\Lambda'$ is a prefix of $\Lambda$.
\end{lemma}
\begin{proof}
    First note that, by a simple intersection argument, there are at most two maximal grade-0 chains $\Lambda_1$ and $\Lambda_2$ at $n$ in step~$s$.

    Next, for every node $n'$ active in step~$s$, let $\Lambda_{n'}$ be the maximal chain with grade~1 at $n'$.
    Note that, by~\Cref{lem:two-thirds-is-maj-correct,lem:maj-correct-is-one-third}, we have that $\Lambda_{n'}$ has grade~0 at $n$ in step~$s$.
    Thus, for every node $n'$ active in step~$s$, $\Lambda_{n'}$ is a prefix of either $\Lambda_1$ or $\Lambda_2$.
    It remains to show that either (a) for every $n'$, $\Lambda_{n'}$ is a prefix of $\Lambda_1$ or (b) for every $n'$, $\Lambda_{n'}$ is a prefix of $\Lambda_2$.

    Consider two correct nodes $n'$ and $n''$ active in step~$s$.
    Suppose towards a contradiction that (a) $\Lambda_{n'}$ is a prefix of $\Lambda_2$ but is incompatible with $\Lambda_1$ and (b) that $\Lambda_{n''}$ is a prefix of $\Lambda_1$ but is incompatible with $\Lambda_2$.
    From (a) and (b) we get that $\Lambda_{n'}$ and $\Lambda_{n''}$ are incompatible.
    This contradicts~\Cref{lem:two-thirds-compatible}, which states that all chains that are the maximal grade-1 chain of some node in step~$s$ are compatible.
\end{proof}

\subsection{Consistency}

First, we show that, in every commit step, maximal grade-0 chains are unique.
This justifies the use of ``the maximal grade-0 chain'' in the commit-step branch of~\Cref{algo:mmr}.
\begin{lemma}
    \label{lem:unique-local-one-third}
    For every $k\geq 0$, in every commit step~$2k+1$, for every node $n$ active in step~$2k+1$, there is a unique maximal grade-0 chain at~$n$.
\end{lemma}
\begin{proof}
    Consider a node $n$ active in commit step~$2k+1$ and suppose towards a contradiction that there are two different chains $\Lambda$ and $\Lambda'$ that are maximal with grade~0 at $n$.
    By the definition of maximal, $\Lambda$ and $\Lambda'$ are incompatible.

    By~\Cref{lem:one-third-includes-correct}, at least one active correct node $n'$ voted for an extension of $\Lambda$ in step~$2k$ and at least one active correct node $n''$ voted for an extension of $\Lambda'$ in step~$2k$.
    If~$k=0$, all correct nodes vote for the empty chain in step~$0$, so $\Lambda$ and $\Lambda'$ are both the empty chain, which is a contradiction.
    If~$k>0$, then step~$2k$ is a proposal step, and by~\Cref{lem:two-thirds-compatible}, all correct nodes that are active in step~$2k$ vote for compatible chains in step~$2k$.
    Thus, $\Lambda$ and $\Lambda'$ are compatible, which is a contradiction.
\end{proof}

Next, we show that, once a chain is committed by a correct node, all correct nodes forever vote for extensions of that chain.
\begin{lemma}
    \label{lem:committed-supported-forever}
    If a chain $\Lambda$ is committed by a correct node in a step~$s$, then, in step~$s$ and in all subsequent steps, all online correct nodes vote for an extension of $\Lambda$.
\end{lemma}
\begin{proof}
    Consider a chain $\Lambda$ committed by a correct node $n$ in a step~$s$, and consider a correct node $n'$ active in step~$s$ and the chain  $\Lambda'$ that $n'$ votes for in step~$s$.

    First, note that, by~\Cref{lem:two-thirds-compatible}, all correct nodes active in step~$s-1$ vote for compatible chains.
    Moreover, by the algorithm, $\Lambda$ has grade~1 at $n$.
    Moreover, by the algorithm, $\Lambda'$ is maximal with grade~0 at $n$.
    Hence, by~\Cref{lem:one-third-extends-two-thirds}, $\Lambda$ is a prefix of $\Lambda'$.

    We have just established that every correct node active in step~$s$ votes in step~$s$ for an extension of~$\Lambda$.
    It is easy to see that, from there on, $\Lambda$ remains a prefix of every vote by every correct node.
\end{proof}

\Cref{lem:committed-supported-forever} easily leads us to our first theorem:
\begin{theorem}
    The MMR algorithm satisfies its consistency property.
\end{theorem}
\begin{proof}
    Consider two chains $\Lambda$ and $\Lambda'$ committed by two correct nodes $n$ and $n'$ in steps~$s$ and~$s'$.
    Note that, by the algorithm, a correct node commits a chain $\Lambda$ only when $\Lambda$ has grade~1.
    Thus, if $s=s'$, then, by~\Cref{lem:two-thirds-compatible}, $\Lambda$ and $\Lambda'$ are compatible.

    Now suppose that $s<s'$.
    By~\Cref{lem:committed-supported-forever}, in step~$s'-1$, all active correct nodes vote for an extension of $\Lambda$.
    Therefore, $\Lambda$ has grade-1 at $n'$ in step~$s'$, and since $\Lambda'$ is the maximal grade-1 chain at $n'$ in step~$s'$, we have that $\Lambda$ is a prefix of $\Lambda'$.
\end{proof}

\subsection{Liveness}

Finally, we turn to liveness.
First we show that, for every proposal step~$2k$, $k\geq 0$, with probability greater than 1/3, there is a correct node~$n$ in step~$2k$ that proposes a chain $\Lambda$ and $\Lambda$ is committed in step~$2k+3$.
\begin{lemma}
    \label{lem:commit-probability}
    For every proposal step~$2k$, $k\geq 0$, with probability greater than 1/3, a chain proposed by a correct node in step~$2k$ is committed in step~$2k+3$.
\end{lemma}
\begin{proof}
    Consider a proposal step~$s=2k$, $k\geq 0$.
    First note that, by assumption, with probability strictly more than 2/3 (by weight), all active correct nodes in step~$s+1$ agree on a leader message $l$ that is sent by a correct node in step~$s$.

    If~$k=0$, all correct nodes vote for the empty chain in step~$0$.
    Thus, the proposal carried by~$l$ extends the maximal grade-0 chain of every active correct node in step~$1$, and every active correct node votes for that proposal in step~$1$.
    The proposal is therefore committed in step~$3$ with probability strictly more than~$2/3$.

    Now assume~$k>0$.
    Moreover, by~\Cref{lem:liveness-key} and by the algorithm, with probability at least 1/2, in step~$s$, the sender of the leader message~$l$ proposes a chain $\Lambda_l$ that is an extension of every chain that every correct node votes for in step~$s$.
    $\Lambda_l$ is therefore an extension of the maximal grade-0 chain of every active correct node in step~$s+1$, and thus, in step~$s+1$, every active correct node votes for $\Lambda_l$.
    $\Lambda_l$ is therefore subsequently committed in step~$s+3$.

    We conclude that, with probability strictly greater than $2/3\cdot 1/2=1/3$, a proposal from a correct node is committed in step~$s+3$.
\end{proof}

\begin{theorem}
    The MMR algorithm satisfies its liveness property.
\end{theorem}
\begin{proof}
    With~\Cref{lem:commit-probability}, we easily obtain the liveness property of total-order broadcast.
\end{proof}

Finally, we show that progress is made in an expected 7 steps despite Byzantine behavior.
\begin{theorem}
    In the most general Byzantine case, for every proposal step~$s$, in expectation, the algorithm commits a block that was proposed by a correct node during or after step~$s$ in step $s+7$.
\end{theorem}
\begin{proof}
    By~\Cref{lem:commit-probability}, successfully committing a block proposed by a correct node is a Bernoulli process with parameter 1/3 and with a trial every 2 steps.
    So, in expectation, the first success happens after 3 trials, {\em i.e.}, in step~$s+4$, and, by the algorithm, the corresponding chain is committed 3 steps later, in step~$s+7$.
\end{proof}

\section{PlusCal/TLA+ Specifications}%
\label{tla-specs}

\subsection{Sieve}

\includegraphics[page=1,width=.9\textwidth]{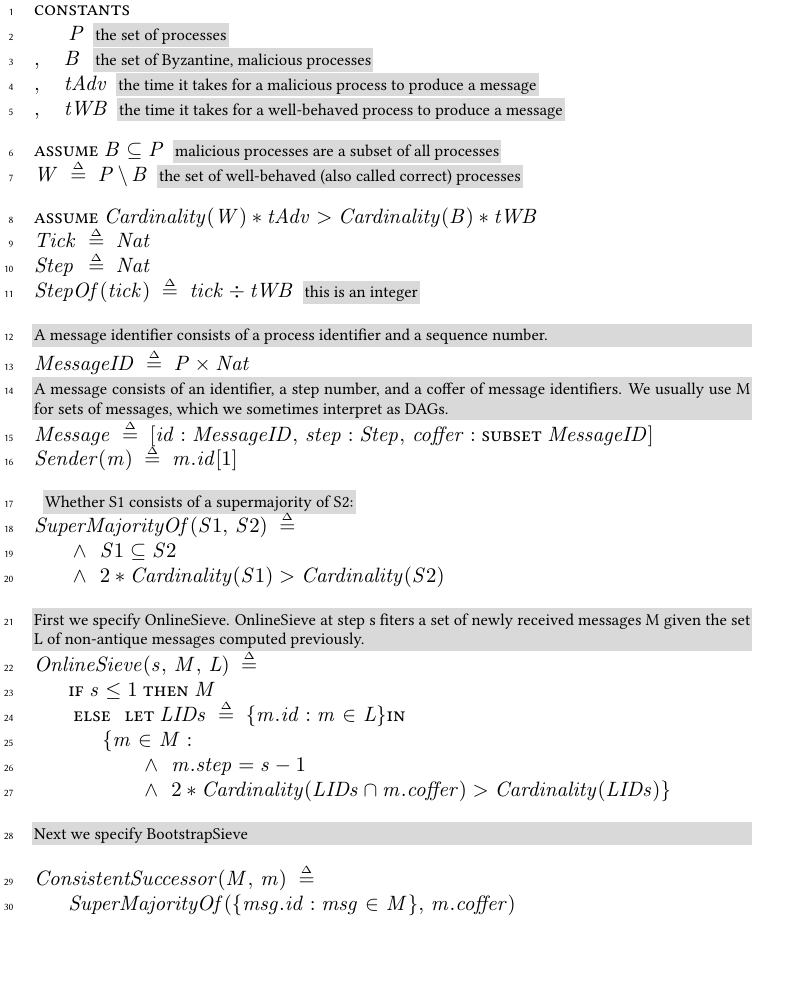}

\includegraphics[page=2,width=.9\textwidth]{Sieve.pdf}

\includegraphics[page=3,width=.9\textwidth]{Sieve.pdf}

\subsection{The MMR Algorithm}

\includegraphics[page=1,width=.9\textwidth]{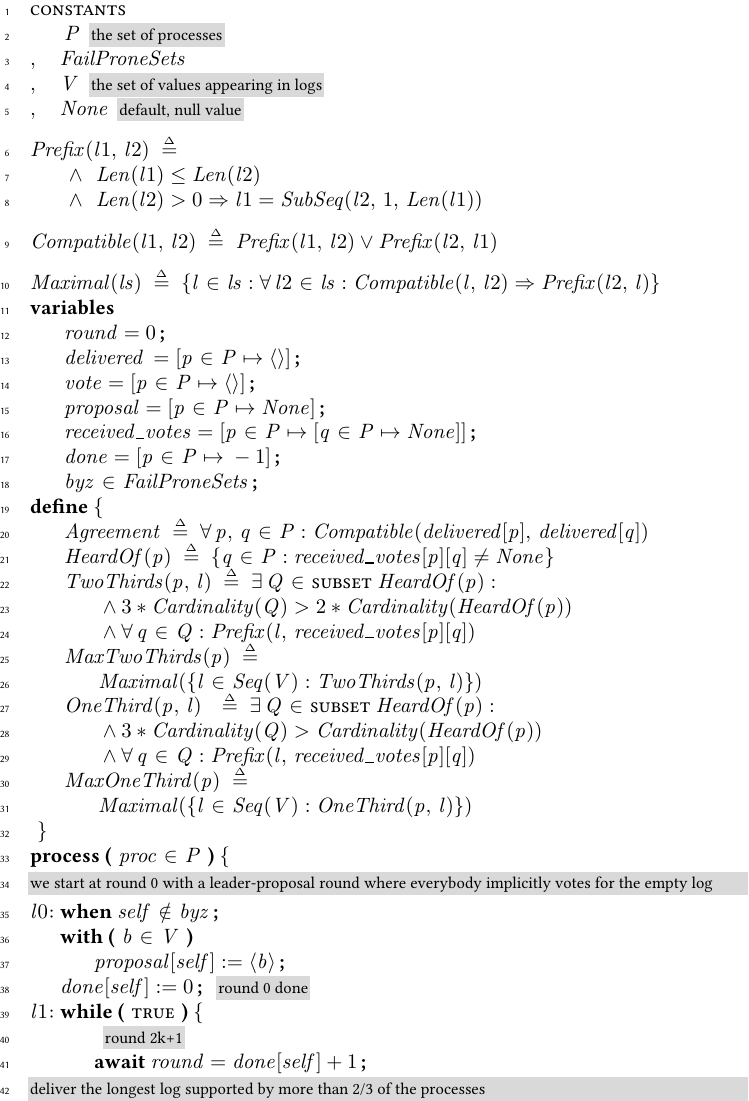}

\includegraphics[page=2,width=.9\textwidth]{MMR.pdf}

\end{document}